\documentclass[twocolumn,twoside]{IEEEtran}
\usepackage{times} 
\usepackage{epsfig}
\usepackage{xtab,diagbox}
\usepackage{arydshln}
\usepackage{graphicx}
\usepackage{epstopdf}
\usepackage{color,subfigure, latexsym, amsmath, amsfonts, amssymb,cite,mathtools}
\usepackage{algorithm}
\usepackage{algorithmic}
\usepackage{epsfig}
\usepackage{graphicx}
\usepackage{epstopdf}
\usepackage{enumitem}
\usepackage{bm,amsmath,amssymb,amsfonts,graphicx,epsfig,amsthm,color, xcolor}
\usepackage{bbold,dsfont}
\usepackage{float}
\usepackage[switch]{lineno}
\usepackage[hyphens]{url}
\PassOptionsToPackage{bookmarks=false}{hyperref}

\usepackage[depth=-1]{bookmark}

\usepackage{booktabs}       
\usepackage{amsfonts}       
\usepackage{microtype}      

\usepackage{times}
\usepackage{graphicx} 

\usepackage{wrapfig}



\usepackage{accents}
\newtheorem{theorem}{Theorem}
\newtheorem{definition}{Definition}
\newtheorem{assumption}{Assumption}
\newtheorem{problem}{Problem}

\theoremstyle{remark}\newtheorem{remark}{Remark}
\theoremstyle{property}\newtheorem{property}{Property}

\allowdisplaybreaks

\title{\LARGE \bf Data-driven Internal Model Control for Output Regulation
}

\author{Wenjie~Liu\textsuperscript{\dag}, Yifei Li\textsuperscript{\dag}, Jian~Sun,~\IEEEmembership{Senior Member,~IEEE}, Gang Wang, Keyou You,\\
 Lihua Xie,~\IEEEmembership{Fellow,~IEEE}, and Jie Chen,~\IEEEmembership{Fellow,~IEEE}
		\thanks{The work was supported in part by the National Natural Science Foundation of China under Grants U23B2059, 62173034, 61925303, and 62088101. (\textsuperscript{\dag}Wenjie Liu and Yifei Li contribute equally to this work.)}
		\thanks{Wenjie Liu, Yifei Li, Jian Sun, Gang Wang, and Jie Chen are with the State Key Lab of Autonomous Intelligent Unmanned Systems, Beijing Institute of Technology, Beijing 100081, China, and also with the Beijing Institute of Technology Chongqing Innovation Center, Chongqing 401120, China (e-mail: liuwenjie@bit.edu.cn; liyifei@bit.edu.cn; sunjian@bit.edu.cn; gangwang@bit.edu.cn; chenjie@bit.edu.cn).

  Keyou You is with the Department of Automation and BNRist, Tsinghua University, Beijing 100084, China (e-mail: youky@tsinghua.edu.cn).

Lihua Xie is with the Centre for Advanced Robotics Technology  Innovation (CARTIN), School of Electrical and Electronic Engineering, Nanyang Technological University, Singapore (e-mail:
elhxie@ntu.edu.sg).
			}
	}

\begin{document}
	
	\maketitle
	

	\begin{abstract}

 Output regulation is a fundamental problem in control theory, extensively studied since the 1970s. Traditionally, research has primarily addressed scenarios where the system model is explicitly known, leaving the problem in the absence of a system model less explored. Leveraging the recent advancements in Willems et al.'s fundamental lemma, data-driven control has emerged as a powerful tool for stabilizing unknown systems. This paper tackles the output regulation problem for unknown single and multi-agent systems (MASs) using noisy data. Previous approaches have attempted to solve data-based output regulation equations (OREs), which are inadequate for achieving zero tracking error with noisy data. To circumvent the need for solving data-based OREs, we propose an internal model-based data-driven controller that reformulates the output regulation problem into a stabilization problem. This method is first applied to linear time-invariant (LTI) systems, demonstrating exact solution capabilities, i.e., zero tracking error, through solving a straightforward data-based linear matrix inequality (LMI). Furthermore, we extend our approach to solve the $k$th-order output regulation problem for nonlinear systems. Extensions to both linear and nonlinear MASs are discussed. Finally, numerical tests validate the effectiveness and correctness of the proposed controllers.
 
	\end{abstract}

\section{Introduction}\label{sec:intro}
The design of a feedback controller to achieve asymptotic tracking for a class of reference inputs and disturbance rejection for a class of disturbances in uncertain dynamical systems, while ensuring closed-loop system stability, is known as the output regulation problem. In this context, both disturbance and reference signals are generated by a known autonomous system, termed the exosystem. This broad mathematical formulation has been applied to numerous real-world control problems, such as the control of unmanned aerial vehicles \cite{huang2004nonlinear,zhou2023uav}, robot arm manipulation \cite{liaquat2016sampled}, and satellite orbiting \cite{shouman2019output}.

The most straightforward solution to this problem involves constructing a controller using the solutions of a set of Sylvester equations, which is known as output regulation equations (OREs) \cite{wonham1974linear}. However, this approach often suffers from limited robustness against model uncertainties. To address this, the internal model principle was introduced in the 1970s through notable works such as \cite{francis1976internal, davison1976robust}, offering an alternative solution to the output regulation problem without the need for solving OREs. Significant research efforts have since been devoted to both linear and nonlinear systems \cite{huang1994robust}. With advancements in computational technology and information science, the application of large-scale systems has become widespread. Consequently, the focus of the output regulation problem has gradually shifted from linear systems to nonlinear systems as well as single systems to multi-agent systems (MASs).

In the MAS context, the classic output regulation problem is often referred to as the cooperative output regulation problem \cite{su2011cooperative}. The objective remains akin to that of a single (linear or nonlinear) system but requires that the strategy be implemented in a distributed manner, respecting the communication graph among agents. This extension introduces several unique challenges, such as switched network topology \cite{su2012cooperative} and communication constraints \cite{tac2020wanglei}. 
Moreover, the results of the cooperative output regulation problem are pivotal in addressing several other fundamental issues in MASs, including the output synchronization problem \cite{wieland2011internal}, and Nash equilibrium seeking \cite{he2023distributed}.

The aforementioned paradigms for solving the output regulation problem are categorized as \emph{model-based} control, relying on accurate system models or requiring \emph{a priori} system identification steps. In contrast, \emph{direct data-driven} control has recently emerged as a new paradigm for situations where modeling complex systems from first principles is challenging, or identifying large-scale systems necessitates extensive data and computational resources. Inspired by the fundamental lemma \cite{willems2005note}, which asserts that the behavior of a linear time-invariant (LTI) system can be linearly expressed in terms of the range space of raw data matrices, a rapidly growing body of direct data-driven control methods has been developed. These methods encompass various applications and generalizations, including robust control \cite{van2020noisy, bisoffi2022data,liu2024learning,wei2024distributed}, model predictive control (MPC) \cite{Coulson2019data, berberich2019data,Liu2022data}, aperiodic control \cite{Wildhagen2023,Qi2023data,li2024self,chen2025data}, nonlinear control \cite{depersis2023cancel, hu2024enforcing,Tao2024nonlinear}, consensus \cite{chang2025localized}, optimization \cite{eising2023cautious}, and control performance analysis \cite{van2022diss,you2023auto}.
Most recently, a data-enabled policy optimization method has been developed in \cite{zhao2024data}, which effectively improves the optimality of the aforementioned data-driven feedback controllers with online closed-loop data.

To date, diverse data-driven control techniques have been developed to address the output synchronization problem \cite{gao2016adaptive,Jiao2021,li2024poly,lin2024distributed}, which is a crucial application of the cooperative output regulation problem. In data-driven contexts, this problem was initially addressed in ideal scenarios. Specifically, \cite{gao2016adaptive,Jiao2021,lin2024distributed} assumes measurable and perfectly known process noise when solving the OREs. Real-world systems, however, often contain pervasive and unmeasurable process noise, compromising the feasibility of OREs. To address this issue, the work \cite{li2024poly} proposed a data-driven polytopic approach providing approximate ORE solutions and stabilizing control gains, achieving near-optimal synchronization. Preliminary data-driven results on the robust output regulation of LTI systems were established in \cite{Zhu2024}, leveraging data informativity theory \cite{vanwaarde2023informativity}. Despite these advances, several research areas remain open, summarized as follows: a1) achieving exact output regulation for unknown LTI systems using noisy data; a2) establishing a direct data-driven output regulation synthesis for unknown nonlinear systems; and, a3) addressing the cooperative output regulation problem for unknown MASs, while ensuring closed-loop stability.

This paper aims to address these open problems step by step. For the first question, previous methods involved reconstructing data-based OREs and devising a controller based on the resulting solution. To overcome the drawback of noise affecting the exact solution to OREs, we adopt the internal model principle, which surprisingly guarantees zero tracking error without solving OREs. Specifically, we design an internal model-based controller using the solution of a data-based linear matrix inequality (LMI), ensuring zero tracking error under noisy data. Building on this idea, we further show that our method solves the $k$th-order output regulation problem for a class of nonlinear systems by utilizing a $k$-fold internal model. Alternative data-driven design methods are discussed, highlighting the generality of our approach. Furthermore, we extend the proposed method to distributed settings for unknown  linear and nonlinear MASs, ensuring successful tracking of an exosystem by all agents.

In summary, the main contribution of this work is threefold:
\begin{itemize}
\item [c1)] Exact output regulation of unknown LTI systems is realized by solving a data-based LMI from noisy input-state data;
\item [c2)] A data-driven controller is designed for unknown nonlinear systems, solving the $k$th-order nonlinear output regulation problem; and, 
\item [c3)] Distributed data-driven control protocols are developed for linear and nonlinear cooperative output regulation problems, with stability guarantees.
\end{itemize}

The paper is organized as follows: Section \ref{sec:pre} provides notation and basic data-driven preliminaries. Section \ref{sec:single} addresses linear output regulation and the $k$th-order nonlinear output regulation problem for unknown systems, and Section \ref{sec:mas}  extends the proposed method to MASs. Section \ref{sec:conclusion} concludes the paper.

\section{Preliminaries}\label{sec:pre}
In this section, we set up the notation and revisit the main result in \cite{bisoffi2022data}, which will be useful throughout the paper.

\subsection{Notation}\label{sec:intro:notation}

We denote the set of real numbers, non-negative integers, and positive integers by $\mathbb{R}$, $\mathbb{N}$, and $\mathbb{N}_{+}$, respectively. The sets of $n$-dimensional real vectors and $n\times m$ real matrices are represented by $\mathbb {R}^n$ and $\mathbb {R}^{n\times m}$, respectively. Additionally, $M \succ (\succeq)~0$ implies that $M$ is positive (semi-)definite, and $M \prec (\preceq)~0$ means that $M$ is negative (semi-)definite. The spectral norm of a matrix $M$ is denoted by $\Vert M\Vert$, and the Euclidean norm of a vector $x \in \mathbb{R}^{n_x}$ is denoted by $\Vert x\Vert$. For a series of column vectors $x_1,\cdots, x_{n_x}$, let ${\rm col}(x_1,\cdots, x_{n_x})$ represent a column vector formed by stacking them, i.e., ${\rm col}(x_1,\cdots, x_{n_x}) = [x_1^\top,\cdots,x_{n_x}^\top]^\top$.

For a constant $\ell \in \mathbb{N}$, let $x^{[\ell]}$ denote the vector containing all monomials of $x$ of degree $\ell$:
\begin{align}
	x^{[\ell]} &= [x_1^\ell,~x_1^{\ell - 1}x_2,~\cdots,~x_1^{\ell - 1}x_{n_x},x_1^{\ell - 2}x_2^2, \nonumber \\
 &~~~\cdots,~x_1^{\ell -2}x_2 x_3,~\cdots,~ x_1^{\ell -2} x_2 x_{n_x},~\cdots,~ x_{n_x}^{\ell}]^\top.
\end{align}
Moreover, for $k = 2, 3, \cdots$, let
\begin{equation}
	x^{(0)} = 1,~~ x^{(1)} = x,~~ x^{(k)} = \underbrace{x\otimes x \otimes \cdots \otimes x}_{k-{\rm factor}}
\end{equation}
where $\otimes$ is the Kronecker product. For matrices $A$, $B$, and $C$ with compatible dimensions, we abbreviate $ABC(AB)^\prime$ to $AB \cdot C(\star)^\prime$. The expression ${\rm diag}\{\cdot\}$ (${\rm blockdiag}\{\cdot\}$) represents a (block) diagonal matrix holding the given elements (matrices) on the main diagonal. The symbol $I_N$ denotes the identity matrix of dimension $N$, and $1_N$ is an $N$-dimensional column vector with all ones.

\subsection{A Robust Data-driven Control Method}\label{sec:pre:ddcontrol}

In the context of data-driven control, it is often assumed that noisy data can be collected beforehand, resulting in a multitude of systems consistent with these data. Therefore, rather than designing a stabilizing controller for a single system as in the model-based approach, the task typically involves designing a stabilizing controller for a set of systems. This is accomplished using robust control methods, such as the S-lemma in \cite{van2020noisy} and Petersen's lemma in \cite{bisoffi2022data}. These methods, known as robust data-driven control methods, are crucial for deriving our main results. Given that different robust data-driven control methods have minimal impact on our design, we briefly review one such method from \cite{bisoffi2022data}.

Consider a continuous-time linear time-invariant (LTI) system
\begin{equation}\label{eq:sys:peter}
\dot{x} = A x + B u
\end{equation}
where $x \in \mathbb{R}^{n_x}$ is the state and $u \in \mathbb{R}^{n_u}$ is the input.
In the data-driven setting, matrices $A$ and $B$ are assumed unknown.
By performing an offline experiment with a  $T$-long input sequence $\{u(\tau)\}_{\tau = 0}^{T - 1}$ to system \eqref{eq:sys:peter}, we can collect a sequence of states $\{x(\tau)\}_{\tau = 0}^{T - 1}$.
For any $\tau \in [t_k, t_{k+1})$, we approximate the state derivative as $\dot{x}(\tau) := 
\frac{x(t_{k + 1}) - x(t_k)}{t_{k + 1} - t_k}$, which satisfies
\begin{equation*}
    \dot{x}(\tau) = A x(\tau) + B u(\tau) + d(\tau),~~\tau \in \{0, \cdots, T - 1\}.
\end{equation*}
where $d\in \mathbb{R}^{n_x}$ 
represents the unknown disturbance, including approximation errors and other noise.

Define the following data matrices
\begin{subequations}\label{eq:data:peter}
\begin{align}
    &U_- = [u(0)~u(1)~\cdots~u(T - 1)]\\
    &X_- = [x(0)~x(1)~\cdots~x(T - 1)]\\
    &X_+ = [\dot{x}(0)~\dot{x}(1)~\cdots~\dot{x}(T - 1)]\\
    &D_- = [d(0)~d(1)~\cdots~d(T - 1)].
\end{align}
\end{subequations}
Due to the uncertainty of $D_-$, a set of system matrices consistent with the data $X_+$, $X_-$, and $U_-$ exists, making it impossible to recover the true matrices $A$ and $B$.
Hence, stabilizing the true system 
$[A~B]$ reduces to stabilizing the set of systems defined by the data. To proceed, we impose the following standard assumptions.
\begin{assumption}\label{as:peter:rank}
The matrix $\left[\begin{smallmatrix} U_-\\
X_- \end{smallmatrix}\right]$ has full row rank.
\end{assumption}
\begin{assumption}\label{as:peter:noise}
The disturbance matrix 
$D_-$
  is bounded; that is, there exists a matrix
    $\Delta$ such that $D_- \in \mathcal{D}$, where  
\begin{equation*}
    \mathcal{D} := \{ D \in \mathbb{R}^{n_x \times T} : DD^\top \preceq \Delta \Delta^\top\}.
\end{equation*}
\end{assumption} 
\begin{remark}
   Assumptions \ref{as:peter:rank} and \ref{as:peter:noise} are very common in the field of data-driven control, as seen in \cite{depersis2021lowcomplexity,bisoffi2022data}. 
   Specifically, Assumption \ref{as:peter:rank} is related to the
notion of persistency of excitation \cite{willems2005note}, implying that the data contains complete information about the  system's dynamics.
Assumption \ref{as:peter:noise} is general enough  to capture several types of disturbances, such as constant, sinusoidal, and truncated Gaussian noise, to name a few.
\end{remark}

Building on these assumptions, the set of all possible matrices $\bar{A}$ and $\bar{B}$ obeying $X_+ = \bar{A} X_- + \bar{B} U_- + D$ with $D \in \mathcal{D}$ is given by
\begin{align}
    \mathcal{C} &:= \{\Phi^\top = [\bar{A}~\bar{B}] : X_+ = \Phi^\top  \left[\begin{matrix}
        X_-\\
        U_-
    \end{matrix}\right] + D,~ D \in \mathcal{D}\}  \nonumber\\
    & = \{\Phi^\top = [\bar{A}~\bar{B}]: \Sigma + \Upsilon^\top \Phi + \Phi^\top \Upsilon + \Phi^\top \Psi \Phi \preceq 0\} \label{eq:peter:setC}
\end{align}
where 
\begin{align}\label{eq:elip:peter:ABC}
	\begin{smallmatrix}\Psi:= \left[\begin{smallmatrix}
		X_-\\
		 U_-
	\end{smallmatrix}\right]\left[\begin{smallmatrix}
	X_-\\
	U_-
	\end{smallmatrix}\right]^\top,~~\Upsilon := - \left[\begin{smallmatrix}
	X_-\\
	U_-
 \end{smallmatrix}\right]X_+^\top,~~\Sigma := X_+ X_+^\top \!-\! \Delta \Delta^\top.\end{smallmatrix}
\end{align}

Using a state feedback controller 
$u = Kx$ for system \eqref{eq:sys:peter}, the following theorem from \cite{bisoffi2022data} provides a stabilizing controller.
\begin{theorem}\label{thm:peter}
    For data matrices $U_-$, $X_-$, $X_+$ in \eqref{eq:data:peter} satisfying Assumptions \ref{as:peter:rank} and \ref{as:peter:noise} and $\Psi$, $\Upsilon$, $\Sigma$ in \eqref{eq:elip:peter:ABC}, the feasibility of the following stabilization problem
    \begin{align}
        &{\rm find}~ K, P = P^\top \succ 0\nonumber \\
        &{\rm s.t.}~ (\bar{A} + \bar{B}K) P + P (\bar{A} + \bar{B}K)^\top \prec 0 ~{\rm for~all}~[\bar{A}~\bar{B}] \in 
        \mathcal{C} \nonumber
    \end{align}
    is equivalent to the feasibility of LMI
    \begin{subequations}\label{eq:sdp:peter}
\begin{align}
        &{\rm find}~Y,P = P^\top \succ 0\\
        &{\rm s.t.}~\left[\begin{matrix}
	- \Sigma &  \Upsilon^\top - \left[\begin{matrix}
		P\\
		Y
	\end{matrix}\right]^\top\\
	\star &  -\Psi
\end{matrix}\right] \prec 0.
    \end{align}
\end{subequations}
    If \eqref{eq:sdp:peter} is solvable, the feedback controller $u = Kx$ with $K = YP^{-1}$ stabilizes system \eqref{eq:sys:peter}.
\end{theorem}

\begin{remark}
Instead of the energy bound in Assumption \ref{as:peter:noise}, other noise bounds can be considered, such as the instantaneous bound
    $d(\tau)d(\tau)^\top \le \delta_{\tau}^2$ in \cite{xie2024data}, or polytopic bounds in \cite{li2024poly}.
    These methods can reduce the size of the set $\mathcal{C}$ in \eqref{eq:peter:setC}, and consequently improve the feasibility of LMI \eqref{eq:sdp:peter}.
\end{remark}

\section{Data-driven Output Regulation}\label{sec:single}

This section addresses the data-driven output regulation problem for both linear and nonlinear systems. We begin by introducing the model-based setup, which serves as a foundation for the data-driven design.

\subsection{Linear Output Regulation}
\subsubsection{Output regulation equations}\label{sec:single:linear}
Consider a continuous LTI system described by
\begin{subequations}\label{eq:sys}
	\begin{align}
		\dot{x} &= A x + B u + E_w w\\
		y &= C x + F_w w
	\end{align}
\end{subequations}
where $x\in \mathbb{R}^{n_x}$ is the state, $u\in \mathbb{R}^{n_u}$ is the control input, $y \in \mathbb{R}^{n_y}$  is the output, and $w \in \mathbb{R}^{n_w}$ is the disturbance.
Following \cite[Chapter 1]{huang2004nonlinear}, the objective of output regulation is to design a controller such that the output  $y(t)$ asymptotically tracks a given reference input $r(t)$, i.e.,
\begin{equation}\label{eq:e}
	\lim_{t \rightarrow \infty} e(t) = \lim_{t \rightarrow \infty} (y(t) - r(t)) = 0.
\end{equation}
Assume that both the disturbance and reference input are generated by linear autonomous differential equations
\begin{subequations}\label{eq:sys:wr}
    \begin{align}
	\dot{w} &= S_{w} w,~~w(0) = w_0\\
	\dot{r} &= S_{r} r,~~r(0) = r_0
\end{align}
\end{subequations}
where $S_{w} \in \mathbb{R}^{n_w \times n_w}$ and $S_{r} \in \mathbb{R}^{n_r \times n_r}$ are assumed known, and $w_0$, $r_0$ are arbitrary initial states. Such formulations (e.g., \eqref{eq:sys:wr}) are general enough, encompassing a broad class of functions, including step, ramp, and sinusoidal functions of various magnitudes and phases.

Due to the similar dynamics of the disturbance and reference signals, they can be considered collectively as the exosignal. Define 
 $n_v := n_w + n_r$ and 
\begin{equation*}
	v = \left[
	\begin{matrix}
		r\\
		w
	\end{matrix}\right] \in \mathbb{R}^{n_v},~~S= \left[\begin{matrix}
	S_{r} & 0\\
	0 & S_{w}
	\end{matrix}\right].
\end{equation*}
The reference inputs and disturbances can be rewritten compactly as follows
\begin{equation}\label{eq:sys:exo}
	\dot{v} = S v,~~v(0) =v_0 = \left[\begin{matrix}
		r_0\\
		w_0
	\end{matrix}\right]
\end{equation}
which we refer to as the exosystem with the exosignal 
 $v$.
The following assumption is imposed.
\begin{assumption}\label{as:exo:lin}
	The matrix $S$ is known and has no eigenvalues with negative real parts.
\end{assumption}
\begin{remark}
    Assumption \ref{as:exo:lin} ensures that the exosignal does not vanish as $t \rightarrow \infty$. Otherwise, the analysis becomes trivial; see \cite[Remark 1.3]{huang2004nonlinear} for details.
\end{remark}

Based on the exosystem above, the system in \eqref{eq:sys} becomes
\begin{subequations}\label{eq:sys:compact}
	\begin{align}
		\dot{x} & = A x + B u + E v \label{eq:sys:compact:x}\\
		e & = C x + F v\label{eq:sys:compact:e}
	\end{align}
\end{subequations}
where $E  = [0~E_w]$ and $F = [-I~F_w]$.
Thus, the linear output regulation problem is formally presented as follows.
\begin{problem}[Linear output regulation]\label{pro:output}
	Design a control law $u$ such that the closed-loop system \eqref{eq:sys:compact:x} with $v(t) = 0$ for all $t\ge 0$ is exponentially stable.
	Additionally, for any initial states $x(0)$ and $v(0)$, the trajectory of \eqref{eq:sys:compact} satisfies \eqref{eq:e}.
\end{problem}
Traditionally, when $v$ is measurable, Problem \ref{pro:output} can be solved using a controller of the form
\begin{equation}\label{eq:ctrl:static}
	u = K_x x + K_v v.
\end{equation}
Let $\bar{x} := x - \Pi v$ for some matrix $\Pi \in \mathbb{R}^{n_x \times n_v}$.
It follows from \eqref{eq:sys:compact} that
\begin{align*}
	\dot{\bar{x}} & =(A + BK_x) x + B K_v v  + E v - \Pi S v\\
	 & = (A + BK_x) \bar{x} + ((A + BK_x)\Pi + BK_v + E - \Pi S)v\\
	 e & = C\bar{x} + (C\Pi + F) v. 
\end{align*}
Thus, if matrices $\Pi$, $K_x$ and $K_v$ are designed such that $A + BK_x$ is Hurwitz stable and the following equations are satisfied
\begin{align*}
	\Pi S & = (A + BK_x)\Pi + BK_v + E \\
	0 & = C\Pi + F 
\end{align*}
then $e(t) \rightarrow 0$ as $t \rightarrow \infty$.
Solving Problem \ref{pro:output} is thus equivalent to finding a controller $u = K_x x + K_v v$ where $K_x$ ensures $A + B K_x$ is Hurwitz stable and $K_v = \Gamma - K_x \Pi$ with $(\Pi, \Gamma)$ obeying
\begin{subequations}\label{eq:outreg2}
    \begin{align}
	\Pi S & = A \Pi + B \Gamma + E \\
	0 & = C\Pi + F.
\end{align}
\end{subequations}
These are known as the output regulation equations (OREs), essential for solving Problem \ref{pro:output}.

In the absence of system models, upon collecting a sequence of noisy input-output-state data, a set of systems consistent with the data exists, i.e., 
$(\bar{A}, \bar{B}, \bar{C}, \bar{E},\bar{F}) \in {\mathcal{C}}$.
Therefore, instead of finding a solution $(\Pi, \Gamma)$ of \eqref{eq:outreg2} for the unique actual system $(A,B,C,E,F)$, we seek a solution $(\Pi, \Gamma)$ such that for all $(\bar{A}, \bar{B}, \bar{C}, \bar{E},\bar{F}) \in {\mathcal{C}}$, the OREs in \eqref{eq:outreg2} hold.
Specifically, we seek  a solution $(\Pi, \Gamma)$ to 
    \begin{align*}
        \left[ \begin{matrix}
            (S^\top \otimes I) - (I \otimes \bar{A}) & (I \otimes -\bar{B})\\
        I \otimes \bar{C} & 0
        \end{matrix}
        \right] 
        \left[\begin{matrix}
            {\rm vec}(\Pi)\\
            {\rm vec}(\Gamma)
        \end{matrix}\right] = \left[\begin{matrix}
            {\rm vec}(\bar{E})\\
            -{\rm vec}(\bar{F})
        \end{matrix}\right]
    \end{align*}
which is infeasible due to the infinitely many systems 
$(\bar{A}, \bar{B}, \bar{C}, \bar{E},\bar{F})$ contained in ${\mathcal{C}}$ compared to the finite number of variables; see \cite[Lemma 1.21]{huang2004nonlinear} for details. A possible solution is to find an approximate solution  $(\bar{\Pi},\bar{\Gamma})$ that minimizes the error of the equations \eqref{eq:outreg2} for all possible matrices $(\bar{A}, \bar{B}, \bar{C}, \bar{E},\bar{F})$, as discussed in \cite{li2024poly}. 
In real-world applications, instead of achieving $\lim_{t \rightarrow \infty}e(t) = 0$, one may only expect $\lim_{t \rightarrow \infty}\|e(t)\| \le \delta$, where $\delta$ is a small constant induced by differences between the actual $(\Pi, \Gamma)$ and the approximated solution  $(\bar{\Pi},\bar{\Gamma})$, as well as the noise magnitude. 

This challenge is similar to recent developments in data-driven LQR and LQG control \cite{depersis2021lowcomplexity,liu2024learning}, where solving algebraic Riccati equations (AREs) is necessary to obtain the LQR and LQG gains. The equivalent LMI formulation of the AREs transforms the problem of solving equations for a set of systems into solving inequalities, providing additional freedom and overcoming the difficulty. A similar approach is expected here to achieve $\lim_{t \rightarrow \infty}e(t) = 0$ by seeking inequalities to replace the OREs in \eqref{eq:outreg2}.

However, unlike AREs, there is no direct LMI formulation for OREs. As shown in \cite[Chapter 1.3]{huang1994robust}, by incorporating an internal model into the original system to create an augmented system, solving OREs can be avoided, and the output regulation problem can be addressed by stabilizing the augmented system. This method transforms the task of solving equations for a set of systems into solving an LMI for a set of systems and avoids using the exosignal $v$, which may not always be measurable.

\subsubsection{Internal model-based method}
To revisit key concepts of the internal model-based method and its data-driven design, consider the following controller applied to the system in  \eqref{eq:sys:compact}
\begin{subequations}\label{eq:ctrl}
	\begin{align}
		u &= K_xx + K_z z\\
		\dot{z} &= G_1 z + G_2 e
	\end{align}
\end{subequations}
where $K_x$, $K_z$ are to be designed, and  $(G_1,G_2)$ is a minimal $n_y$-copy internal model of $S$ defined as follows.
\begin{definition}[Internal model \cite{huang2004nonlinear}]\label{def:internal}
	Given any square matrix $S$, a pair of matrices $(G_1, G_2)$ incorporates an $n_y$-copy internal model of $S$ if 
	\begin{equation}
		G_1 = {\rm blockdiag}(\underbrace{\beta,\cdots,\beta}_{n_y-{\rm tuple}}),~G_2 = {\rm blockdiag}(\underbrace{\sigma,\cdots,\sigma}_{n_y-{\rm tuple}})
	\end{equation}
where $\beta$ is a constant square matrix whose characteristic polynomial equals the minimal polynomial of $S$, and $\sigma$ is a constant column vector such that $(\beta, \sigma)$ is controllable.
\end{definition} 

Based on this definition, the closed-loop system \eqref{eq:sys} with the controller \eqref{eq:ctrl} is given by
\begin{subequations}\label{eq:sys:interlti}
	\begin{align*}
		\dot{x} & = A x + B u + E v\\
  e & = C x + F v\\
  u &= K_x x + K_{z} z\\
  \dot{z} &= G_1 z + G_2 e\\
  \dot{v} & = S v.
	\end{align*}
\end{subequations}
Letting $\xi = {\rm col}(x,z) \in \mathbb{R}^{n_{\xi}}$ with $n_{\xi} := n_x + n_{z}$, we obtain
\begin{subequations}\label{eq:sys:interlti1}
	\begin{align}
		\dot{\xi} & = A_\xi \xi + B_\xi u + E_\xi v\\
		e &= C_\xi \xi + F v\\
  \dot{v} & = S v
	\end{align}
\end{subequations}
where 
	\begin{align*}
		A_\xi &= \left[\begin{matrix}
			A & 0\\
			G_2 C & G_1
		\end{matrix}\right],~B_\xi = \left[\begin{matrix}
		B\\
		0
		\end{matrix}\right],\\
  E_\xi &= \left[\begin{matrix}
		E\\
		G_2 F
		\end{matrix}\right], ~C_\xi = \left[\begin{matrix}
		C&0
		\end{matrix}\right].
	\end{align*}
It has been shown in \cite[Lemma 1.27]{huang2004nonlinear} that if the controller gain $K_\xi := [K_x~K_z]$ is designed such that $A_\xi + B_\xi K_\xi$ is Hurwitz stable, then a solution to the OREs \eqref{eq:outreg2} for the actual system $(A,B,C,E,F)$ exists and can be constructed by $(A,B,C,E,F,K_x, K_z)$.
This indicates that the controller \eqref{eq:ctrl} solves Problem \ref{pro:output}.
Thus, addressing the OREs translates into designing a stabilizing gain matrix $K_\xi$, achievable by solving LMIs. This method is generalized to the data-driven setting in the following section.

\subsubsection{Data-driven output regulation}\label{sec:single:linear:data}
To design the output regulation controller in the absence of a system model, offline experiments are conducted to collect noisy input-state data
\begin{equation*}
	\mathbb{D} := \{(u(\tau), {\xi}(\tau),\dot{\xi}(\tau)\}_{\tau = 0}^{T- 1}
\end{equation*}
from the system
\begin{align*}
	\dot{\xi} & = A_\xi \xi + B_\xi u + E_\xi v + d\\
	\dot{v} & = S v\\
	e &= C_\xi \xi + F v
\end{align*}
where $d$ is the disturbance in data collection as defined in Section \ref{sec:pre:ddcontrol}.
Rearrange the data to form matrices
\begin{subequations}\label{eq:data:lti}
	\begin{align}
		U_- & = \left[\begin{matrix}
			{u}(0) & {u}(1) & \cdots & {u}(T-1)
		\end{matrix}\right]\\
  \Xi_- &= \left[\begin{matrix}
			{\xi}(0) & {\xi}(1) & \cdots & {\xi}(T-1)
		\end{matrix}\right]\\
  \Xi_+ &= \left[\begin{matrix}
			\dot{\xi}(0) & \dot{\xi}(1) & \cdots & \dot{\xi}(T-1)
		\end{matrix}\right].
	\end{align}
\end{subequations}
Define the data matrix of the unknown disturbance $d$ and the exosignal $v$ as 
\begin{align*}
	D &= \left[\begin{matrix}
		{d}(0) & {d}(1) & \cdots & {d}(T-1)
	\end{matrix}\right]\\
	V &= \left[\begin{matrix}
		{v}(0) & {v}(1) & \cdots & {v}(T-1)
	\end{matrix}\right].
\end{align*}
Assumptions \ref{as:peter:rank} and \ref{as:peter:noise} are rephrased as follows.
\begin{assumption}\label{as:rank:lti}
	The  matrix  $\left[\begin{smallmatrix}
		U_-\\
		\Xi_-
	\end{smallmatrix}\right]$ has full row rank.
\end{assumption}
\begin{assumption}\label{as:noise:lin}
    The matrices $V$ and $D$ have bounded energy, i.e., there exists some matrix $\Delta$ such that  
\begin{equation}\label{eq:noisebound}
	(E_\xi V + D)(E_\xi V + D)^\top \preceq \Delta \Delta^\top.
\end{equation}
\end{assumption}
Building on the results in Section \ref{sec:pre:ddcontrol}, 
it follows that $[A_\xi, B_\xi] \in \mathcal{C}_\xi$, where
\begin{equation}
	\mathcal{C}_\xi := \{\bar{\Phi}_\xi^\top = [\bar{A_\xi}~\bar{B}_\xi]: \Sigma + \Upsilon^\top \bar{\Phi}_\xi + \bar{\Phi}_\xi^\top \Upsilon + \bar{\Phi}_\xi^\top \Psi \bar{\Phi}_\xi \preceq 0\}
\end{equation} 
with 
\begin{align}\label{eq:elip:ABC}
	\begin{smallmatrix}\Psi := \left[\begin{smallmatrix}
		\Xi_-\\
		 U_-
	\end{smallmatrix}\right]\left[\begin{smallmatrix}
	\Xi_-\\
	U_-
	\end{smallmatrix}\right]^\top,~~\Upsilon := - \left[\begin{smallmatrix}
	\Xi_-\\
	U_-
	\end{smallmatrix}\right]\Xi_+^\top,~~\Sigma := \Xi_+ \Xi_+^\top - \Delta \Delta^\top.\end{smallmatrix}
\end{align}
To address Problem \ref{pro:output} in the absence of a system model, it is necessary to determine a controller gain 
$K_\xi$
  that stabilizes all pairs of system matrices $[\bar{A}_\xi~\bar{B}_\xi] \in \mathcal{C}_\xi$.
This requirement is satisfied by the following theorem. \begin{theorem}\label{thm:LTI}
	Under Assumption \ref{as:exo:lin}, for data matrices $U_-$, $\Xi_-$, $\Xi_+$ in \eqref{eq:data:lti} that satisfy Assumptions \ref{as:rank:lti} and \ref{as:noise:lin}, and $\Psi$, $\Upsilon$, $\Sigma$ in \eqref{eq:elip:ABC}, if the LMI in \eqref{eq:lmi:linear} is feasible for some matrices $P \succ 0$ and $Y$, then the controller \eqref{eq:ctrl} with $K_\xi = YP^{-1}$ stabilizes the augmented system \eqref{eq:sys:interlti1}.
	Furthermore, the controller \eqref{eq:ctrl} solves Problem \ref{pro:output}.
	\begin{equation}\label{eq:lmi:linear}
\left[\begin{matrix}
	- \Sigma &  \Upsilon^\top - \left[\begin{matrix}
		P\\
		Y
	\end{matrix}\right]^\top\\
	\star &  -\Psi
\end{matrix}\right] \prec 0.
	\end{equation} 
\end{theorem}
\begin{proof}
  Theorem \ref{thm:peter} establishes that the feasibility of   \eqref{eq:lmi:linear} implies that the resulting gain matrix $K_{\xi}$ ensures 
    \begin{equation}
        (\bar{A}_{\xi} + \bar{B}_{\xi} K_{\xi})P + P(\bar{A}_{\xi} + \bar{B}_{\xi} K_{\xi})^\top - P \prec 0 
    \end{equation}
    for all $[\bar{A}_{\xi}~\bar{B}_{\xi}] \in \mathcal{C}_{\xi}$.
    Since $[A_{\xi}~B_{\xi}] \in \mathcal{C}_{\xi}$, this further indicates that $K_{\xi}$ ensures that $A_{\xi} + B_{\xi} K_{\xi}$ is Hurwitz stable, completing the proof according to \cite[Chapter 1.3]{huang2004nonlinear}.
\end{proof}

The effectiveness of the proposed method is demonstrated through a numerical example below. All simulations were performed using Matlab 2022a on a Lenovo laptop with a 14-core i7-12700H processor at 2.3GHz. The proposed LMIs were solved using CVX \cite{grant2014cvx}.

\subsubsection{Example 1}\label{sec:single:linear:exm}
Consider the dynamics of a robot system originally studied in \cite{su2012general}. The system matrices 
$A$, $B$, $C$, $E$, $F$, and $S$ are given by
\begin{align*}
    &A = \left[\begin{matrix}
        0 & 1\\
        1 & 2
    \end{matrix}\right], B = \left[\begin{matrix}
        0\\
        1
        \end{matrix}\right],C= [1~0],E = \left[\begin{matrix}
            0 & 0 & 0 & 0\\
            1 & 0 & 0 & 0
        \end{matrix}\right],\\
        & F = \left[
        \begin{matrix}
           -1&0&-1&0 
        \end{matrix}
        \right],~S = \left[\begin{matrix}
    0 & \omega_1 & 0 & 0\\
    -\omega_1 & 0 & 0 & 0\\
    0 & 0 & 0 & \omega_2\\
    0 & 0 & -\omega_2 & 0
\end{matrix}\right]
\end{align*}
where $\omega_1 = \pi/5$ and $\omega_2 = 1$.
According to Definition \ref{def:internal}, the matrices $G_1$ and $G_2$ were chosen as follows
\begin{equation}\label{eq:g1g2}
    G_1 = \left[
    \begin{matrix}
    0 & 1& 0& 0\\
    0 &0 &1 &0\\
    0 &0& 0& 1\\
    -0.3948 & 0 & -1.3948&  0
    \end{matrix}
    \right],
G_2 = \left[
    \begin{matrix}
    0\\
    0\\
    0\\
    1
    \end{matrix}
    \right].
\end{equation}
We collected noisy data trajectories of length $T = 20$ from random initial conditions with random inputs uniformly generated from $[-0.5,0.5]$, the exosignal $v$ from $[-0.0025,0.0025]$ and the noise $d$ from $[-0.01,0.01]$.
Upon solving the LMI \eqref{eq:lmi:linear} and designing the controller as in \eqref{eq:ctrl}, Fig. \ref{fig:lin} illustrates the tracking performance under the proposed controller. The tracking error is observed to asymptotically converge to zero, verifying the correctness and effectiveness of the proposed data-driven control method.

\begin{figure}
	\centering
	\includegraphics[width=9cm]{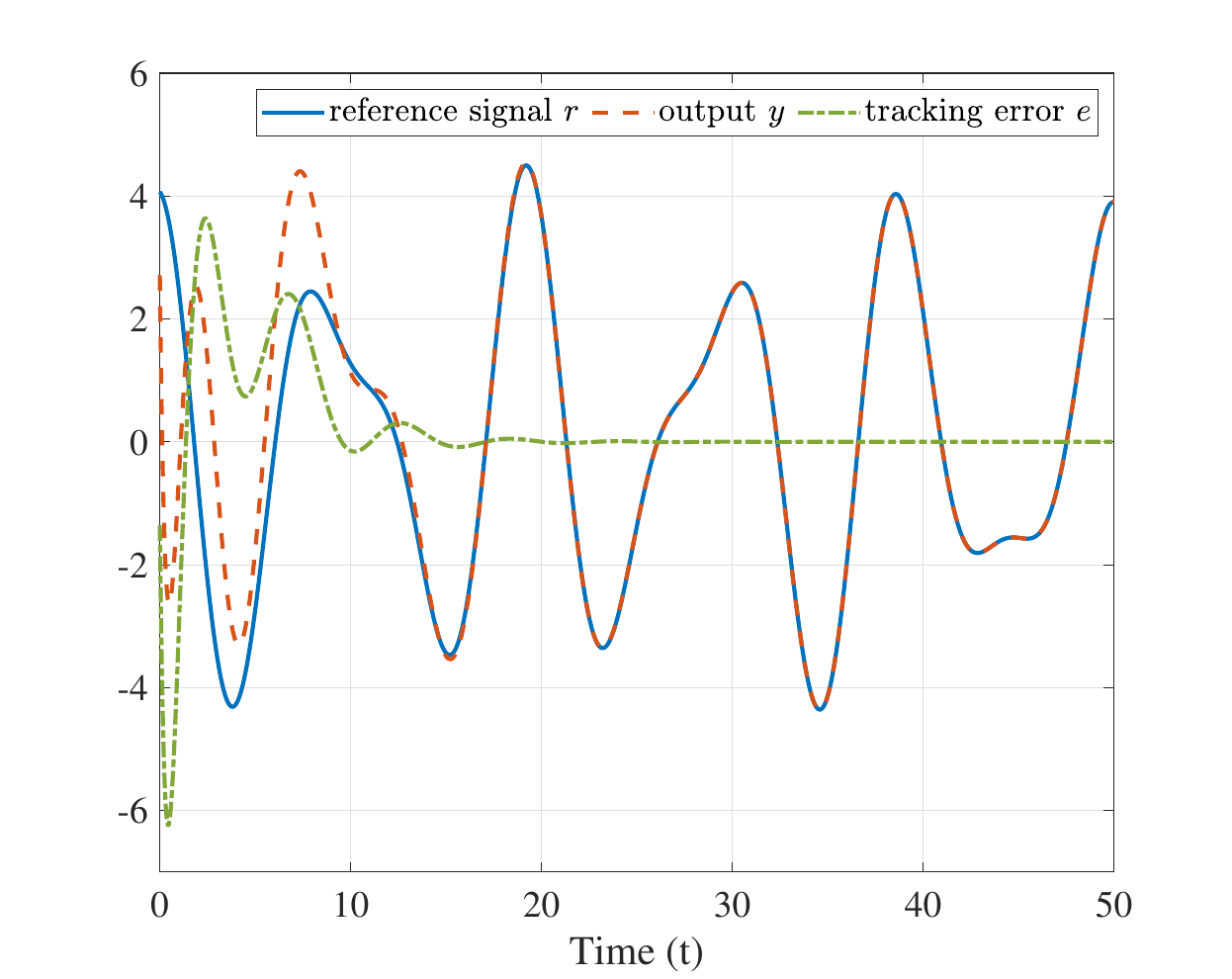}\\
	\caption{Tracking performance under the controller \eqref{eq:ctrl}.}\label{fig:lin}
	\centering
\end{figure}

\subsection{Nonlinear Output Regulation}\label{sec:single:non}
In the previous section, we presented an internal model-based data-driven controller for the linear output regulation problem. This approach suggests the possibility of developing a similar design for the nonlinear output regulation problem. Since achieving exact nonlinear output regulation, i.e., zero tracking error, is challenging even in the model-based scenario, we introduce a more tractable problem, termed the 
$k$th-order nonlinear output regulation problem. We first outline its model-based design and then extend the method to devise a data-driven controller.

\subsubsection{Nonlinear Output Regulation Equations} \label{sec:single:non:equations}
Consider a smooth nonlinear system
\begin{subequations}\label{eq:sys:non0}
	\begin{align}
		\dot{x} & = f(x,u,v) \\
		\dot{v} & = S v \label{eq:sys:non0:exo} \\
		e &= h(x,v)
	\end{align}
\end{subequations}
where the vectors $x$, $u$, $v$, and $e$ hold the same meanings as in the linear case in Section \ref{sec:single:linear}.
Assume that $(x^*, u^*, 0)$ is an equilibrium of the function $f(x,u,v)$.
Given that any known equilibrium can be transformed to the origin by a change of coordinates, we assume without loss of generality that $(x^*, u^*) = (0,0)$.
Additionally, the exosystem satisfies the following assumption.
\begin{assumption}\label{as:exo:non}
	The matrix $S$ in \eqref{eq:sys:non0:exo} is known and all its eigenvalues have zero real parts.
\end{assumption}
Similar to the linear case, addressing the nonlinear output regulation problem typically requires solving the following nonlinear output regulation equations (NOREs)
\begin{subequations}\label{eq:outreg:non}
	\begin{align}
		\frac{\partial \mathbf{x}(v)}{\partial v} S v & = f(\mathbf{x}(v), \mathbf{u}(v), v)\\
		0 & = h(\mathbf{x}(v),v)
	\end{align}
\end{subequations}
where $\mathbf{x}(v)$ and $\mathbf{u}(v)$ are two sufficiently smooth functions defined for $v$ satisfying $\mathbf{x}(0) = 0$ and $\mathbf{u}(0) = 0$.
The controller can then be designed based on  $\mathbf{x}(v)$ and $\mathbf{u}(v)$.
However, solving the NOREs \eqref{eq:outreg:non} is very challenging even in the model-based scenario. Therefore, instead of seeking exact output regulation, we consider a more tractable problem known as the 
$k$th-order nonlinear output regulation problem \cite{huang1994robust}, described as follows.

\begin{problem}\label{pro:k-output}
{\textbf{\textup{($k$th-order nonlinear output regulation)}}}
	Design a control law $u$ such that the closed-loop system \eqref{eq:sys:non0} satisfies the following two properties.
\end{problem}

\begin{property}\label{proper:nonstable}
For all sufficiently small $x(0)$ and $v(0)$, the trajectory ${\rm col}(x(t), v(t))$ of the closed-loop system \eqref{eq:sys:non0} exists and is bounded for all $t \ge 0$.
\end{property}
\begin{property}\label{proper:koutreg}
	For all sufficiently small $x(0)$ and $v(0)$, the trajectory ${\rm col}(x(t), v(t))$ of the closed-loop system \eqref{eq:sys:non0} adheres to
 \small\begin{equation*}
		\lim_{t \rightarrow \infty} (e(t) \!- O(v^{(k+1)}(t)) \!= \lim_{t \rightarrow \infty} (h(x(t),v(t)) - O(v^{(k+1)}(t))) = 0
	\end{equation*}
 \normalsize
	where $O(v^{(k+1)})$ is such that 
\begin{equation}\label{eq:property:o}
		\lim_{v \rightarrow 0} \frac{\Vert O(v^{(k+1)})\Vert}{\Vert v\Vert^{k+1}}
	\end{equation}
 is a finite constant.
\end{property}

With these definitions, finding functions  $\mathbf{x}(v)$ and $ \mathbf{u}(v)$ that solve the NOREs in \eqref{eq:outreg:non} boils down to finding functions $\mathbf{x}^{k}(v)$ and $\mathbf{u}^{k}(v)$ with $\mathbf{x}^{k}(0) = 0,\mathbf{u}^{k}(0) = 0$ such that
\begin{subequations}\label{eq:koutreg:non}
    \begin{align}
        \frac{\partial \mathbf{x}^{k}(v)}{\partial v} S v &= f(\mathbf{x}^{k}(v), \mathbf{u}^{k}(v), v)\\
		o^k(v) & = h(\mathbf{x}^{k}(v), \mathbf{u}^{k}(v),v).
    \end{align}
\end{subequations}
These equations are referred to as the $k$th-order NOREs. It has been shown in \cite[Lemma 4.7]{huang2004nonlinear} that solutions to equations \eqref{eq:koutreg:non}, i.e., $\mathbf{x}^{k}(v)$ and $\mathbf{u}^{k}(v)$, can be represented using the Taylor's series expansion of the functions $f$ and $h$.

Clearly, this approach is infeasible when only noisy input-state data are available, as the accurate Taylor's series expansions of functions  $f$ and $h$ cannot be explicitly determined. Therefore, inspired by the method in Section \ref{sec:single:linear},  we seek an alternative approach to address Problem \ref{pro:k-output} without solving the $k$th-order NOREs in \eqref{eq:koutreg:non}.
This can be achieved by integrating the original system with a well-designed internal model, as described in \cite{huang1994robust}. We will first briefly revisit this method and then present its data-driven design.

\subsubsection{Internal Model of the $k$-Fold Exosystem}\label{sec:single:non:inter}

We begin by introducing some necessary notation. Leveraging Taylor’s expansion, the first-order approximation of \eqref{eq:sys:non0} around the equilibrium $(0,0,0)$ is given by
\begin{subequations}\label{eq:sys:non}
	\begin{align}
		\dot{x} & = f(x,u,v) = A x + B u + E v + \alpha(x,u,v)\\
		\dot{v} & = S v\\
		e &= h(x,v) = C x + F v + \gamma(x,v). 
	\end{align}
\end{subequations}
Here, the functions $\alpha(x,u,v)$ and $\gamma(x,v)$ represent the higher-order remainders, and the matrices 
 $A$, $B$, $C$, $E$, and $F$ are defined as
\begin{align*}
    A &= \frac{\partial f}{\partial x} (0,0,0), ~B = \frac{\partial f}{\partial u}(0,0,0),~C = \frac{\partial h}{\partial x} (0,0,0)\\
    E &= \frac{\partial f}{\partial v}(0,0,0), ~F = \frac{\partial h}{\partial v}(0,0,0).
\end{align*}

Using the notation $v^{[\ell]}$ and $v^{(\ell)}$ from Section \ref{sec:intro:notation}, let $M_{\ell}$ and $N_{\ell}$ be constant matrices such that 
\begin{equation*}
	v^{[\ell]} = M_{\ell} v^{(\ell)},~~v^{(\ell)} = N_\ell v^{[\ell]}.
\end{equation*}
According to \cite[Lemma 1]{huang1994robust}, we have that $\frac{d}{dt}\dot{v}^{[\ell]} = S^{[\ell]} v^{[\ell]}$, where
\begin{equation}\label{eq:S^ell}
	S^{[\ell]} = M_{\ell }\Big[\sum_{i = 1}^{\ell} I_{n_v}^{i-1} \otimes S \otimes I_{q^{\ell - i}}\Big]N_{\ell}.
\end{equation}

Building on these preliminaries, we use the same controller as for the linear case
\begin{subequations}\label{eq:ctrl:non}
	\begin{align}
		u & = K_x x + K_z z\\
		\dot{z} & = G_1 z + G_2 e.
	\end{align}
\end{subequations}
Instead of incorporating an $n_y$-copy internal model of the original exosystem $\dot{v} = Sv$ from Section \ref{sec:single:linear}, the pair $(G_1, G_2)$ in \eqref{eq:ctrl:non} incorporates an $n_y$-copy internal model of the $k$-fold exosystem, that is,
\begin{equation}
    \frac{d}{dt}\left[
    \begin{matrix}
        v^{[1]}\\
        v^{[2]}\\
        \vdots\\
        v^{[k]}
    \end{matrix}
    \right] = \mathcal{S}_{kf} \left[
    \begin{matrix}
        v^{[1]}\\
        v^{[2]}\\
        \vdots\\
        v^{[k]}
    \end{matrix}
    \right]
\end{equation}
where
\begin{equation}\label{eq:S_kf}
	\mathcal{S}_{kf} = \left[
	\begin{matrix}
		S^{[1]} & 0 & \cdots & 0\\
		0 & S^{[2]} & \cdots & 0\\
		\vdots & \vdots & \ddots & \vdots\\
		0 & 0 & \cdots & S^{[k]}
	\end{matrix}
	\right]
\end{equation} 
and $S^{[\ell]}$ is given by \eqref{eq:S^ell} for $\ell = 1,2,\cdots,k$.
This internal model is known as the $k$th-order internal model.
Plugging this controller into the system \eqref{eq:sys:non} and letting $\xi := {\rm col}(x, z)$, we obtain that 
\begin{subequations}\label{eq:sys:non:aug}
	\begin{align}
\left[
\begin{matrix}
    \dot{x}\\
    \dot{z}
\end{matrix}
\right]	& = {\left[
\begin{matrix}
    A & 0\\
    G_2 C & G_1
\end{matrix}
\right]} {\left[
\begin{matrix}
   x\\
    z
\end{matrix}
\right]} + {\left[
\begin{matrix}
    B\\
    0
\end{matrix}
\right]} u + {\left[
\begin{matrix}
    E\\
    G_2 F
\end{matrix}
\right]} v + {\left[
\begin{matrix}
    \alpha(x,u,v) \\
    G_2 \gamma(x,v)
\end{matrix}
\right]}\nonumber\\
& \!\!:= A_{\xi} \xi + B_{\xi} u + E_\xi v + \psi(\xi,v)\\
	\dot{v} & = S v\\
	u & = [K_x ~K_z] \xi := K_{\xi} \xi\\
	e &= [C~0] \xi + F v + \gamma(x,v).
\end{align}
\end{subequations}

It has been shown in \cite[Theorem 3.12]{huang1994robust} that as long as the controller gain matrix 
$K_{\xi}$ is designed such that $A_\xi + B_\xi K_\xi$ is Hurwitz, Problem \ref{pro:k-output} is solved.
Therefore, Theorem \ref{thm:LTI} can be directly extended to solve the 
$k$th-order nonlinear output regulation problem.

\subsubsection{A Data-Driven Implementation} \label{sec:single:non:data}

To gather information about the nonlinear system \eqref{eq:sys:non:aug}, $T$ independent offline experiments are performed around the equilibrium, resulting in the data $\mathbb{D} := \{u(\tau), {\xi}(\tau),\dot{\xi}(\tau)\}_{\tau = 0}^{T- 1}$.
Forming the matrices $\Xi_+$, $\Xi_-$ and $U_-$ as in \eqref{eq:data:lti}
\begin{subequations}\label{eq:data:non}
	\begin{align}
	U_- & = \left[\begin{matrix}
			{u}(0) & {u}(1) & \cdots & {u}(T-1)
		\end{matrix}\right]	\\
  \Xi_- &= \left[\begin{matrix}
			{\xi}(0) & {\xi}(1) & \cdots & {\xi}(T-1)
		\end{matrix}\right]\\
  \Xi_+ &= \left[\begin{matrix}
			\dot{\xi}(0) & \dot{\xi}(1) & \cdots & \dot{\xi}(T-1)
		\end{matrix}\right].
	\end{align}
\end{subequations}
These matrices are assumed to contain sufficient information about the nonlinear system \eqref{eq:sys:non0}, in accordance with Assumption \ref{as:rank:lti}. Likewise, let us define the data matrices of the unknown higher-order remainders and the exosignal as follows
\begin{align}
	D &= \left[\begin{matrix}
		\psi(\xi(0),v(0)) & \cdots & \psi(\xi(T-1),v(T-1))
	\end{matrix}\right]\\
	V &= \left[\begin{matrix}
		{v}(0) & {v}(1) & \cdots & {v}(T-1)
	\end{matrix}\right]
\end{align}
which are assumed to be bounded, i.e., obeying Assumption \ref{as:noise:lin}.

\begin{remark}
The approximation error between the actual derivative and the approximated derivative, denoted as
 $d$, is neglected as it plays a role analogous to the approximation error $\psi(x,u)$. As analyzed in \cite[Remark 8]{guo2022data}, the term 
 $\psi(\xi(\tau),u(\tau), v(\tau))^\top \psi(\xi(\tau),u(\tau), v(\tau)) \le \sum_{\iota = 1}^{n_x} (n_x + n_u + n_v)L_{\iota}^2\Vert (\xi(\tau), u(\tau), v(\tau))\Vert$ with $L_{\iota} >0$ being some Lipschitz constant. This indicates that during data collection, the closer the input-state data are to the equilibrium, the smaller $\Delta$ in \eqref{eq:noisebound} will be. 
\end{remark}

Based on these data, we extend Theorem \ref{thm:LTI} to nonlinear systems. The proof follows directly from \cite[Theorem 3.12]{huang1994robust} and that of Theorem \ref{thm:LTI}, and is omitted here.

\begin{theorem}\label{thm:non}
	Under Assumption \ref{as:exo:non}, for data $U_-$, $\Xi_-$, $\Xi_+$ in \eqref{eq:data:non} satisfying Assumptions \ref{as:rank:lti} and \ref{as:noise:lin} and $\Psi$, $\Upsilon$, $\Sigma$ in \eqref{eq:elip:ABC}, if the LMI in \eqref{eq:lmi:non} is feasible for some matrices $P \succ 0$ and $Y$, then the  controller \eqref{eq:ctrl:non} with $K_\xi = YP^{-1}$ solves Problem \ref{pro:k-output}
	\begin{equation}\label{eq:lmi:non}
\left[\begin{matrix}
	- \Sigma &  \Upsilon^\top - \left[\begin{matrix}
		P\\
		Y
	\end{matrix}\right]^\top\\
	\star &  -\Psi
\end{matrix}\right] \prec 0.
	\end{equation} 
\end{theorem}

Regarding this theorem, we provide the following remarks.

\begin{remark}[Nonlinear exosystem]
   The presented method is applicable to nonlinear exosystems, i.e., $\dot{v} = S(v)$ where $S(v)$ is a nonlinear function of $v$; see \cite[Chapter 3]{huang2004nonlinear} for details.

\end{remark}

\begin{remark}[Choice of controller \eqref{eq:ctrl:non}]\label{rmk:nonctrl}
It is worth noting that the data-driven controller design for Problem \ref{pro:k-output} is not unique. For illustration and convenience, we adopt the same controller as in the linear case. By embedding an internal model of the $k$-fold exosystem into the original system, the output regulation problem is transformed into the stabilization problem of the augmented system with $v = 0$, as follows
\begin{equation}\label{eq:augsys}
    \left[\begin{matrix}
    \dot{x}\\
    \dot{z}
\end{matrix}\right] = \left[\begin{matrix}
    f(x,z,0)\\
    G_1 z + G_2 e
\end{matrix}\right].
\end{equation}
Therefore, as long as the controller $u = k(x,z)$ stabilizes the augmented system \eqref{eq:augsys}, Problem \ref{pro:k-output} is addressed.
For instance, if we consider $f(x,u,v)$ having the form $f(x,u,v) = A Z( x ) + B u + E_v$ with $Z(x)$ containing both the linear $x$ as well as the vector of nonlinear functions of $x$, then the controller can be designed as 
\begin{equation}\label{eq:ctrl:cancel}
    u = K_x Z(x) + K_z z = K_\xi \left[\begin{matrix}
    Z(x)\\z
\end{matrix}\right].
\end{equation}
In this case, several data-driven methods can be applied, e.g., \cite{depersis2023cancel,hu2024enforcing}, which can guarantee the global stability of the augmented system under mild conditions. However, the global stability of the augmented system \eqref{eq:sys:non:aug} with $v = 0$ does not imply global output regulation \cite[Remark 7.1]{huang2004nonlinear}. Global output regulation for nonlinear systems remains an open question.
\end{remark}

\subsubsection{Example 2} \label{sec:single:non:exm}
Consider the dynamics of the ball and beam system, adapted from \cite[Chapter 5.6]{huang2004nonlinear}
\begin{subequations}
	\begin{align}
		\dot{x}_1 &= x_2 + v_2\\
		\dot{x}_2 &= h_0 x_1 x_4^2 - g h_0 \sin{x_3}\\
  \dot{x}_3 &= x_4\\
  \dot{x}_4 &= u\\
  e & = x_1 - (v_1 + v_3)
	\end{align}
\end{subequations}
where $h_0 = 0.7134$, $g = 9.81$, and the exosystem is the same as  that in Section \ref{sec:single:linear:exm}. 

Noisy trajectories were collected with a number of $T = 50$ independent experiments from random initial conditions and random inputs uniformly generated from $[-0.1,0.1]$, the exosignal $v$ from $[-0.005,0.005]$, and noise $d$ from $[-0.002,0.002]$. By implementing a $2$nd-order internal model, the nonlinear tracking performance of the proposed controller \eqref{eq:ctrl:non} is shown in Fig. \ref{fig:non}(a).
In addition, Fig. \ref{fig:non}(b) compares the tracking performances of the controller \eqref{eq:ctrl:non} with a 1st-order internal model (blue solid line), a $2$nd-order internal model (red dashed line), and the controller 
\eqref{eq:ctrl:cancel} with  a $2$nd-order internal model with $K_\xi$ designed using the method in \cite{depersis2023cancel} (green dash-dotted line).

Furthermore, considering system \eqref{eq:sys:non:aug} with $v = 0$, Fig. \ref{fig:non:v}(a) illustrates that both the state and the output converge to zero under the controller \eqref{eq:ctrl:non}.
However, using the same controller with the same initial condition $\xi(0)$, it can be observed from Fig. \ref{fig:non:v}(b) that the tracking error diverges for sufficiently large $v(0)$. This indicates that, in the nonlinear case, a controller stabilizing the system \eqref{eq:sys:non:aug} with $v = 0$  does not imply output regulation.
This differs from the linear case and verifies Remark \ref{rmk:nonctrl}.

\begin{figure}[!htb]
	\centering
	\includegraphics[width=9cm]{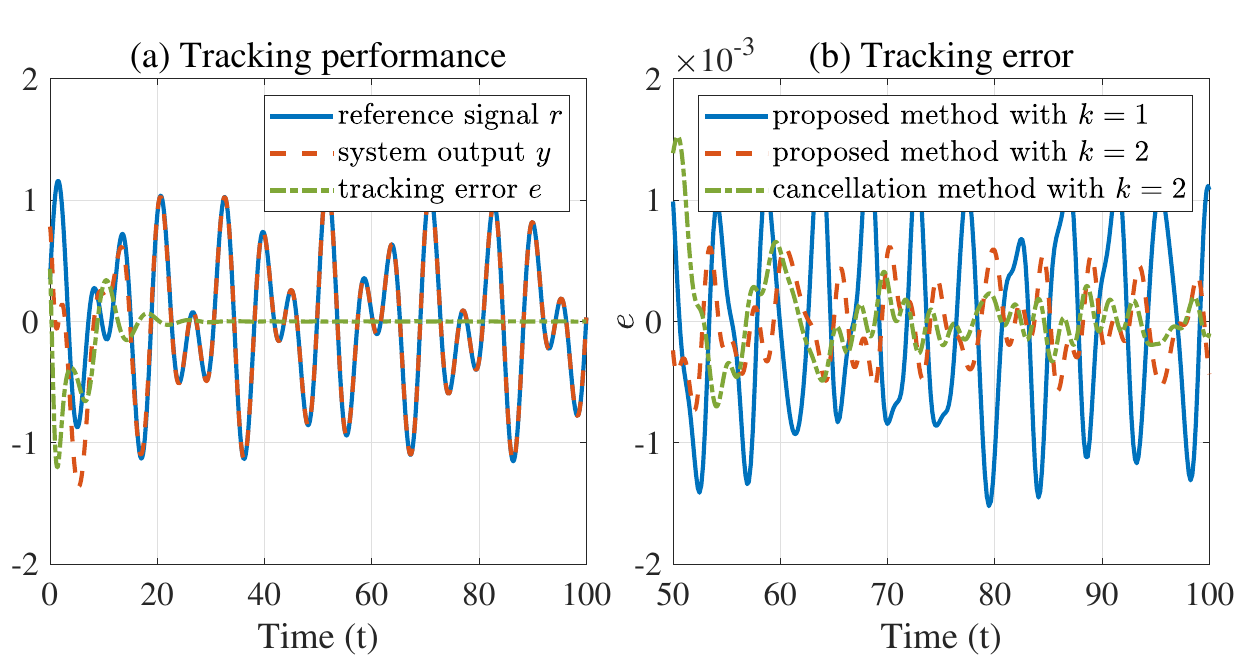}\\
	\caption{Data-driven output regulation of nonlinear systems.}\label{fig:non}
	\centering
\end{figure}

\begin{figure}[!htb]
	\centering
	\includegraphics[width=9cm]{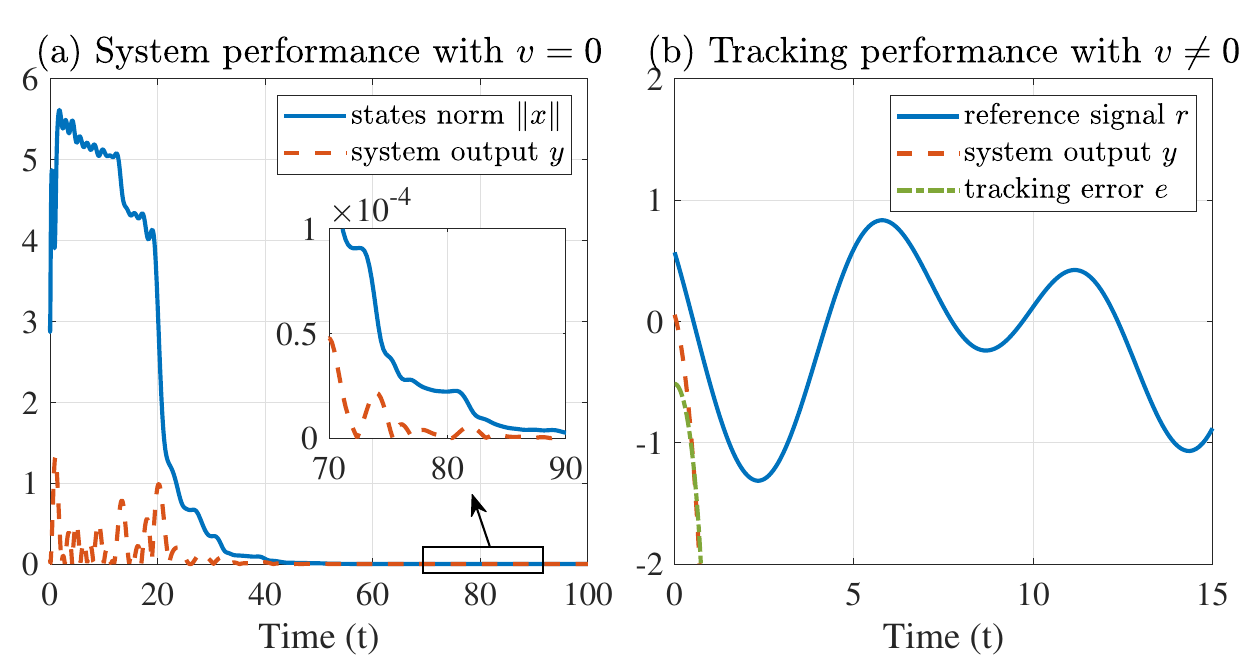}\\
\caption{Data-driven local output regulation of nonlinear systems.}\label{fig:non:v}
	\centering
\end{figure}

\section{Data-driven Cooperative Output Regulation}\label{sec:mas}

In the previous sections, we proposed data-driven solutions for output regulation problems of both linear and nonlinear time-invariant systems. In this section, we extend these solutions to address linear and nonlinear multi-agent systems (MASs).

\subsection{Linear Multi-Agent Systems}

Consider a continuous-time linear MAS composed of 
 $N$ heterogeneous agents 
\begin{equation}
	\begin{split}
		\label{eq:mas}
		\dot{x}_{i}&= {A}_{i} x_{i}+{B}_{i} u_{i}+E_{i} v\\
		y_i&={C}_ix_{i},\quad i=1,2,\ldots,N
	\end{split}
\end{equation}
where $x_i\in \mathbb{R}^{n_{xi}}$, $u_i\in \mathbb{R}^{n_{ui}}$, and $y_i \in \mathbb{R}^{n_{y}}$ represent the state,  control input, and  output of agent $i$, respectively. The matrices ${A}_i$, ${B}_i$, ${C}_i$, and $E_i$ are fixed but unknown. Recalling Section 
\ref{sec:single:linear}, $v \in\mathbb{R}^{n_v}$ consists of the reference signal to be tracked and the disturbance to be rejected, assumed to be generated by the following exosystem
\begin{equation}
	\label{eq:mas:exo}
	\begin{split}
		\dot{v} &={S} v \\
		{y}_0 &=-{F} v 
	\end{split}
\end{equation}
where $y_{0} \in \mathbb{R}^{n_{y}}$ represents the output of the exosystem. The exosystem satisfies Assumption~\ref{as:exo:lin}.
We define the tracking error of agent $i$ as $e_i = y_i -y_0 $.

It is worth noting that the dimensions of the dynamics and/or states can differ across the 
 $N$ agents, while their output dimensions must be identical to achieve cooperative output regulation. The objective is to synchronize the outputs of all agents to that of the exosystem by implementing distributed controllers locally at each agent such that
\begin{equation}\label{eq:mas:goal}
\lim_{t\rightarrow\infty}e_i(t)=\lim_{t\rightarrow\infty}(y_i(t)-y_0(t))=0
\end{equation}
holds for all $i=1,2,\ldots,N$.

To analyze and synthesize the cooperative output regulation problem of linear MASs, it is instrumental to revisit some concepts in graph theory.

\emph{(Graph theory.)} Consider a weighted graph ${\mathcal{G}}=({\mathcal{V}}, {\mathcal{E}})$ to depict the communication topology among $N$ agents in \eqref{eq:mas}. Here, ${\mathcal{V}}=\{\nu_1,\ldots, \nu_N \}$ denotes a nonempty set of nodes, while ${\mathcal{E}} \subseteq {\mathcal{V}}\times {\mathcal{V}}$ represents a set of edges. The edge $(\nu_i, \nu_j)$ belongs to ${\mathcal{E}}$ if there is a link from node $\nu_j$ to node $\nu_i$. The adjacency matrix $\mathcal{A}=[a_{ij}]\in \mathbb{R}^{N\times N}$ is defined such that $a_{ij}>0$ if $(\nu_j, \nu_i)\in \mathcal{E}$, and $a_{ij}=0$ otherwise. It is assumed that there are no self-loops, meaning that $a_{ii}=0$ holds for all $i=1,2,\ldots,N$. Let $\mathcal{L}=[l_{ij}]\in \mathbb{R}^{N\times N}$ denote the Laplacian of ${\mathcal{G}}$ corresponding to $\mathcal{A}$, where $l_{ii} = \Sigma_{j=1}^{N}a_{ij}$ and $l_{ij} = -a_{ij}$ for $i\neq j$.

In the context of cooperative output regulation, the 
 $N$ agents in \eqref{eq:mas} along with the exosystem in \eqref{eq:sys:exo} can be collectively viewed as a leader-following MAS, with the exosystem as the leader and the $N$ agents as the followers. The interactions within the leader-following MAS are modeled by an extended graph $\bar{\mathcal{G}}=(\bar{\mathcal{V}}, \bar{\mathcal{E}})$, where $\bar{\mathcal{V}}={\mathcal{V}}\cup \nu_0$ with $\nu_0$ representing  the exosystem node. The set $\bar{\mathcal{E}}$ includes all the arcs in ${\mathcal{E}}$ as well as the arcs between $\nu_0$ and ${\mathcal{E}}$. 
 
A graph $\bar{\mathcal{G}}$ is said to contain a directed spanning tree if there exists a node, known as the root, from which every other node in $\bar{\mathcal{V}}$ can be reached through a directed path. The pinning matrix $\Lambda = {\rm{diag}}\{a_{10},\ldots,a_{N0}\}\in\mathbb{R}^{N\times N}$ describes the accessibility of the node $\nu_0$ to the remaining nodes $\nu_i\in \mathcal{V}$. Specifically, $a_{i0} > 0$ if $(\nu_0,\nu_i)\in \bar{\mathcal{E}}$, and $a_{i0} = 0$ otherwise. Define the matrix $H:=\mathcal{L}+\Lambda$. Then, we have $H\mathbf{1}_N =\Lambda\mathbf{1}_N$. 
Denote all the eigenvalues of the matrix
${H}$ by $\lambda_i$ for $i=1,2,\ldots,N$.

Before proceeding, a standard assumption about the communication topology for cooperative output regulation is provided as follows.

 \begin{assumption}
	\label{as:graph}
	The graph $\bar{\mathcal{G}}$ contains a directed spanning tree with the node $\nu_0$ as the root.
\end{assumption}

Now, we introduce a typical distributed control protocol. First, we define a virtual tracking error for each agent 
 $i$ as follows 
\begin{equation}
	\label{eq:virtual:error}
		e_{vi} =\sum_{j=1}^{N}a_{ij}(y_i -y_j )+a_{i0}(y_i -y_0 ).
\end{equation}
Consider a distributed state feedback control law for each agent in \eqref{eq:mas} as follows
\begin{subequations}\label{eq:controller}
    \begin{align}
		u_{i} &=K_{xi}\Big(\sum_{j=1}^{N}a_{ij}(x_j -x_i )+a_{i0}x_i \Big)+K_{zi}z_i \label{eq:controller:u}\\
		\dot {z}_i  &= G_1z_i +G_2e_{vi} , \quad i=1,2,\ldots,N	\label{eq:controller:z}	
    \end{align}
\end{subequations}
where $z_i \in\mathbb{R}^{n_{z}}$ and  $K_{zi}$, $K_{xi}$, $G_1$, and $G_2$ are constant matrices to be designed. In particular, the matrix pair $(G_1,G_2)$ is defined as in Definition~\ref{def:internal}.

Define $x ={\rm col}(x_1,\ldots,x_N)$, $z ={\rm col}(z_1 ,\ldots,z_N)$, $u ={\rm col}(u_1 ,\ldots,u_N)$, $e_v ={\rm col}(e_{v1},\ldots,e_{vN})$, $\tilde{A}={\rm{blockdiag}}\{A_{1},\ldots,A_{N}\}$, $\tilde{B}={\rm{blockdiag}}\{B_{1},\ldots,B_{N}\}$, $\tilde{C}={\rm{blockdiag}}\{C_{1},\ldots,C_{N}\}$,  $\tilde{E}=[E_{1}^\top,\ldots,E_{N}^\top]^\top$, and $\tilde F=\Lambda 1_N\otimes F$.
For the entire network, we can define an augmented system as follows
\begin{subequations}\label{eq:mas:close:x}
\begin{align}
		\dot x  &=  \tilde Ax +\tilde B u + \tilde E v \\
		\dot v  &= S v \\
		e_v  &= (H\otimes I_p)\tilde{C}x +\tilde F v \\
  \dot {z}  &= (I_N\otimes G_1)z + (I_N\otimes G_2)e_{v} .
\end{align}
\end{subequations}

Let $\xi :={\rm col}(x,z)\in \mathbb{R}^{N n_{\xi i}}$ with $n_{\xi i}=n_{xi}+n_{z}$. Define 
$\tilde K_x={\rm{blockdiag}}\{K_{x1},\ldots,K_{xN}\}$ and
$\tilde K_z={\rm{blockdiag}}\{K_{z1},\ldots,K_{zN}\}$.
For the distributed control protocol \eqref{eq:controller}, the closed-loop system is given by
\begin{subequations}\label{eq:mas:close}
    \begin{align}
    	\dot \xi  &=  \tilde A_{\xi}  \xi + \tilde E_{\xi}v  \label{eq:mas:close1}\\
		\dot v  &= S v \\
		e_v  &= (H\otimes I_p)\tilde{C}x +\tilde F v
  \end{align}
\end{subequations}
where 
\small
\begin{equation*}
	 \tilde A_{\xi}=\left[\begin{matrix}
		\tilde A+(H\otimes I_n)\tilde B\tilde K_x & \tilde B\tilde K_z\\
		(H\otimes G_2)\tilde C &I_N\otimes G_1
	\end{matrix}
	\right],  \bar E_{\xi}=\left[\begin{matrix}
		\tilde E\\
		 (I_N\otimes G_2)\tilde F
	\end{matrix}
	\right].
\end{equation*}
\normalsize

\begin{problem}
\label{pro:coop:linear}{\textbf{\textup{(Linear cooperative output regulation)}}}
Given the MAS \eqref{eq:mas}, the exosystem \eqref{eq:mas:exo}, and a diagraph $\bar{\mathcal{G}}$, design a distributed control law of the form \eqref{eq:controller} such that \eqref{eq:mas:goal} holds for all $x_i(0)$, $v(0)$, and $i=1,2,\ldots,N$.
\end{problem}

It has been shown in \cite[Lemma 1]{su2012cooperative} that under Assumption~\ref{as:graph}, $\lim_{t\rightarrow\infty}e_{i}(t) =0$ if and only if $\lim_{t\rightarrow\infty}e_{vi}(t) =0$.
In line with this, the distributed control protocol \eqref{eq:controller} will solve the linear cooperative output regulation problem (cf. Problem \ref{pro:coop:linear}), if the gain matrix $K_{\xi}:={\rm blockdiag}\{K_{\xi 1},\ldots,K_{\xi N}\}$ with each element $K_{\xi i}:=[K_{xi}, K_{zi}]$ is designed such that the closed-loop system \eqref{eq:mas:close} is asymptotically stable, i.e., the matrix $\tilde{A_\xi}$ is Hurwitz stable.
However, complexities surge when the system matrices $(A_i,B_i,C_i,E_i)$ are unknown.
The challenge we face is to address Problem~\ref{pro:coop:linear} directly from  data.

In this pursuit, similar to Section \ref{sec:single:linear:data}, we collect a set of data $\mathbb{D}_i := \{(u_i(\tau),{\xi}_i(\tau),\dot{\xi}_i(\tau)\}_{\tau = 0}^{T- 1}$ through an offline experiment on the perturbed system
\begin{subequations}\label{eq:mas:offline}
   \begin{align}
    \label{eq:mas:data}
    \dot{\xi}_{i} &= {A}_{\xi i} \xi_{i} +{B}_{\xi i} u_{i} +E_{\xi i}v +d_i \\
  \dot v & = S  v 
\end{align} 
\end{subequations}
where $\xi_i ={\rm col}(x_i,z_i)$, $A_{\xi i}=\left[\begin{smallmatrix}
		A_i & 0 \\
		G_2 C_i & G_1
	\end{smallmatrix}
	\right]$, $B_{\xi i}=\left[\begin{smallmatrix}
		B_i\\
		0 
	\end{smallmatrix}
	\right]$, $E_{\xi i}=\left[\begin{smallmatrix}
		E_i\\
		 G_2F
	\end{smallmatrix}
	\right]$,
and $d_i \in\mathbb{R}^{n_{\xi i}}$ represents unknown disturbance during data collection.

To store the collected data, we define the following matrices per agent
\begin{subequations}\label{eq:mas:datamatrix}
  \begin{align}
	U_{i-}& := \left[\begin{matrix}u_i(0) & u_i(1) & \cdots & u_i(T-1)\end{matrix}\right]\\
	\Xi_{i-}& :=\left[\begin{matrix}
		\xi_i(0) &\xi_i(1) & \cdots& \xi_i(T-1)	
	\end{matrix}\right]\\
	\Xi_{i+} &:=\left[\begin{matrix}
		\dot \xi_i(0) &\dot \xi_i(1) & \cdots& \dot \xi_i(T-1)
	\end{matrix}\right].
\end{align}  
\end{subequations}

Note that these matrices are related by the equation
\begin{equation}\label{eq:relation}
    \Xi_{i+}=	A_{\xi i}	\Xi_{i-}+ 	B_{\xi i}U_{i-}+E_{\xi i}V+D_i
\end{equation}
where $V:=[v(0) \; v(1) \; \cdots \; v(T-1)]$ and $D_i:=[d_i(0) \; d_i(1) \; \cdots \; d_i(T-1)]$ are unknown matrices of agent $i$. To further the design and analysis, we impose some requirements on the data, introducing the following assumptions.

\begin{assumption}\label{as:rank}
    For all $i=1,2,\ldots,N$, the data matrices satisfy 
    ${\rm rank}\!\left(\left[\begin{smallmatrix}
		U_{i-}\\
		\Xi_{i-}
	\end{smallmatrix}\right]\right) = n_{\xi i}$. 
\end{assumption}

\begin{assumption}\label{as:noise}
The sequences $V$ and $D_i$ are bounded, i.e., there exists
    \begin{equation}
	\label{lift_noise}
	(E_{\xi i} V + D_i)(E_{\xi i}  V + D_i)^\top \preceq \Delta_i \Delta_i^\top
    \end{equation}
for some matrix $\Delta_i$ and $i=1,2,\ldots,N$.
\end{assumption}

The set of all system matrices consistent with the data is defined as
\small
\begin{equation}\label{eq:mas:elip}
    \mathcal{C}_i:=\left\{\bar \Phi_i^\top=[\bar A_{\xi i}\;\bar B_{\xi i}]:
\Sigma_i + \Upsilon_i^\top \Phi_i + \Phi_i^\top \Upsilon_i + \Phi_i^\top \Psi_i \Phi_i \preceq 0\right\}
\end{equation} 
\normalsize
where $\Psi_i := \left[\begin{smallmatrix}
		\Xi_{i-}\\
		 U_{i-}
	\end{smallmatrix}\right]\left[\begin{smallmatrix}
	\Xi_{i-}\\
	U_{i-}
	\end{smallmatrix}\right]^\top$, $\Upsilon_i := - \left[\begin{smallmatrix}
	\Xi_{i-}\\
	U_{i-}
	\end{smallmatrix}\right]\Xi_{i+}^\top$, and $\Sigma_i := \Xi_{i+} \Xi_{i+}^\top - \Delta_i \Delta_i^\top$ for $i=1,2,\ldots,N$.

Therefore, the proposed distributed control protocol  \eqref{eq:controller}, which solves Problem \ref{pro:coop:linear} without the knowledge of system models, is established by the following theorem.

\begin{theorem}\label{thm:mas:linear}
    Consider the MAS \eqref{eq:mas}, the exosystem \eqref{eq:mas:exo}, and the graph $\bar{\mathcal{G}}$ under Assumptions~\ref{as:exo:lin} and \ref{as:graph}. 
For data $U_{i-}$, $\Xi_{i-}$, $\Xi_{i+}$ in \eqref{eq:mas:datamatrix} satisfying Assumptions \ref{as:rank} and \ref{as:noise} and $\Psi_i$, $\Upsilon_i$, $\Sigma_i$ in \eqref{eq:mas:elip}, 
 if the LMIs in \eqref{eq:mas:sdp1} are feasible for some matrices $P_i \succ 0$ and $Y_i$, $i=1,2,\ldots,N$, then the matrix $K_{\xi}$ with $K_{\xi i}:= Y_iP_i^{-1}/\lambda_i$ renders $\tilde{A}_\xi$ Hurwitz stable.
Furthermore, the distributed control protocol \eqref{eq:controller} solves Problem \ref{pro:coop:linear}.
	\begin{equation}\label{eq:mas:sdp1}
\left[\begin{matrix}
	- \Sigma_i &  \Upsilon_i^\top - \left[\begin{matrix}
		P_i\\
		Y_i
	\end{matrix}\right]^\top\\
	\star &  -\Psi_i
\end{matrix}\right]\prec 0.
	\end{equation} 
\end{theorem}

\begin{proof}
Under Assumption~\ref{as:graph}, all the eigenvalues $\lambda_i$ for $i=1,2,\ldots,N$ have positive real parts.
We represent $T_1\in\mathbb{R}^{N\times N}$ as a unitary matrix such that $J_{{H}}=T_1{H}T_1^{-1}{=\rm diag}\{\lambda_1,\lambda_2,\ldots,\lambda_N\}$.
Let $T_2=\left[\begin{smallmatrix}
	T_1\otimes I_n & 0 \\
	 0  &T_1\otimes I_{n}
\end{smallmatrix}
\right]$. 
Then, we have that
$$\hat A_{\xi}:=T_2\tilde A_{\xi}T_2^{-1}=\left[\begin{matrix}
	\tilde A+(J_{{H}}\otimes I_n)\tilde B\tilde K_x & \tilde B\tilde K_z\\
(J_{{H}}\otimes G_2)\tilde C &I_N\otimes G_1
\end{matrix}
\right].$$
Furthermore, let $T_3={\rm col}(i_1,i_{N+1},i_2,i_{N+2},\ldots,i_N,i_{2N})$, where $i_k$ is the $k$th row of $I_{2N}$.
It deduces that $\check A_{\xi}:=(T_3\otimes I_n)\hat A_{\xi}(T_3^{-1}\otimes I_n)$ is a lower block triangular matrix whose diagonal blocks are
$$\check A_{\xi i}:=\left[\begin{matrix}
 A_i+\lambda_iB_iK_{xi} &  B_iK_{zi}\\
	\lambda_iG_2C_i &G_1
\end{matrix}
\right] \text{for}~i=1,2,\ldots,N.$$
Let $T_{4i}=\left[\begin{smallmatrix}
	I_n &  0 \\
	0  &\lambda_i^{-1}I_{n}
\end{smallmatrix}
\right]$.
We obtain that
$$\underline {A}_{\xi i}:=T_{4i}\check A_{\xi i}T_{4i}^{-1} =\left[\begin{matrix}
	A_i+\lambda_iB_iK_{xi} &  	\lambda_iB_iK_{zi}\\
G_2C_i &G_1
\end{matrix}
\right].$$
Thus, we conclude that $\tilde A_{\xi}$ is Hurwitz stable if and only if, for all $i=1,2,\ldots,N$, $\check A_{\xi i}$ and hence $\underline{A}_{\xi i}$ are Hurwitz.

Considering the aforementioned,  similar to Theorem~\ref{thm:LTI}, the feasibility of \eqref{eq:mas:sdp1} implies the resultant gain matrices $K_{\xi i}$ satisfy
\begin{equation}
   (\bar A_{\xi i} + \bar B_{\xi i}\mathcal{K}_i)P_i+P_i(\bar A_{\xi i} + \bar B_{\xi i}\mathcal{K}_i)^\top\prec 0
\end{equation}
for all $(\bar A_{\xi i}, \bar B_{\xi i})\in\mathcal{C}_i$. 
Since the true matrices $(A_{\xi i}, B_{\xi i})\in \mathcal{C}_i$, we deduce that $ A_{\xi i} + \lambda_iB_{\xi i}{K}_{\xi i}$ are Hurwitz stable for all $i=1,2,\ldots,N$, 
which implies $\tilde{A}_{\xi}$ is Hurwitz stable.

Next, it follows from \cite[Lemma 1.27]{huang2004nonlinear} that there exists a unique $\tilde X_{\xi}=[\Pi_1^\top,\cdots,\Pi_N^\top,\Gamma_1^\top,\cdots,\Gamma_N^\top]^\top$ that satisfies
\begin{subequations}
    \begin{align}
    \tilde X_{\xi}S&=\tilde A_{\xi}\tilde X_{\xi}+\tilde E_{\xi}\\
    0&=\tilde C_{\xi}\tilde X_{\xi}+F
\end{align}
where $\tilde C_{\xi}:=[\tilde C\quad  {0}]$.
\end{subequations}
This completes the proof.
\end{proof}

\begin{remark}[Comparison]
The data-driven output synchronization problem, which is a special case of the cooperative output regulation problem (corresponding to 
$E=0$), has been studied in \cite{Jiao2021} and \cite{li2024poly}. 
Compared with these existing results, the proposed method has the following advantages.
\begin{itemize}
    \item [1)]  
    \emph{Less conservative assumptions on noise.} The approach presented in \cite{Jiao2021} requires that the process noise during offline data collection is measurable and perfectly known, which is impractical in real-world applications. In contrast, our proposed method only assumes the noise is bounded, i.e., satisfying Assumption \ref{as:noise}.
    \item [2)] 
    \emph{Zero tracking error.} The approach presented in \cite{li2024poly} only establishes the ultimately uniformly bounded tracking error due to the infeasibility of the OREs. In contrast, our proposed dynamic control law avoids solving OREs and achieves zero tracking error, i.e., 
    $\lim_{t\rightarrow\infty}e_i(t)=0$ for $i=1,2,\ldots,N$.
    \item[3)]
    \emph{Addressing the nonlinearity.} The works \cite{lopez2024,Jiao2021,li2024poly}  only consider the linear case, whereas our proposed method can deal with nonlinear systems, which allows for a wider range of application scenarios.
\end{itemize}

\end{remark}

We illustrate the aforementioned results via a numerical example.

\emph{Example 3:} Consider a MAS consisting of four robot systems modeled in Section~\ref{sec:single:linear:exm}, treating each system as agent $i$ for $i=1,\ldots,4$.
The system matrices of each agent are given by
\begin{align*}
    &A_i = \left[\begin{matrix}
        0 & 1\\
        1-i & 2-i
    \end{matrix}\right], B_i = B,C_i= C, E_i = \left[\begin{matrix}
            0 & 0 & 0 & 0\\
            i & 0 & 0 & 0
        \end{matrix}\right].
\end{align*}
The matrices $B$ and $C$, as well as the exosystem, are chosen to be the same as those in Section \ref{sec:single:linear:exm}. 
The information exchange among all agents is described by the digraph $\bar{\mathcal{G}}$ in Fig.~\ref{fig:graph}, where the node $0$ represents the exosystem $v(t)$. It is evident that Assumption~\ref{as:graph} holds.

\begin{figure}[!htb]
	\centering
	\includegraphics[width = 5cm]{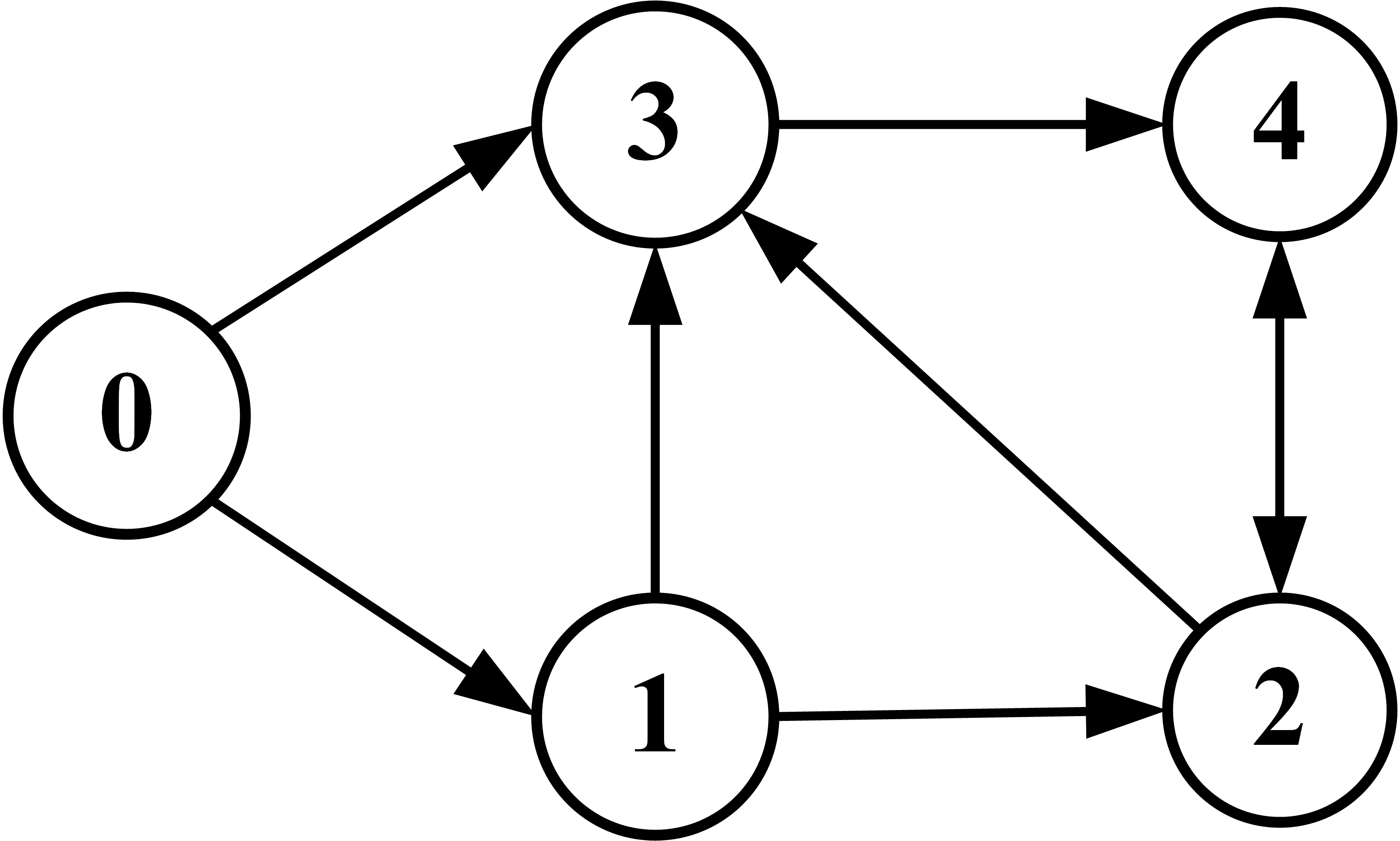}\\
	\caption{The communication graph $\bar{\mathcal{G}}$ between the $4$ agents and the exosystem}.\label{fig:graph}
	\centering
\end{figure}

We collected each agent's trajectories of length $T = 20$ from the perturbed system \eqref{eq:mas:offline} under the same conditions  as in Example 1, assuming that the noise $d_i$ is bounded by $[-0.01,0.01]$ for all $i=1,\dots,4$.
By solving the LMIs \eqref{eq:mas:sdp1} in Theorem~\ref{thm:mas:linear}, we obtain that the controller gain $K_{\xi i}$ for each agent.
Fig. \ref{fig:mas:lin} illustrates the tracking performance of the MAS under the proposed distributed data-driven control protocol \eqref{eq:controller}. We observe that the output of each agent asymptotically tracks the output of the exosystem, indicating that the cooperative output regulation problem is successfully solved.

In addition, Fig. \ref{fig:mas:linpoly} compares the tracking error $e_i(t)$ under the proposed method and the polytopic method in \cite{li2024poly}.
It is evident that the proposed method achieves exact tracking, whereas the polytopic method only provides bounded tracking, thereby demonstrating the superiority of the proposed approach.

\begin{figure}[!htb]
	\centering
	\includegraphics[width = 9.5cm]{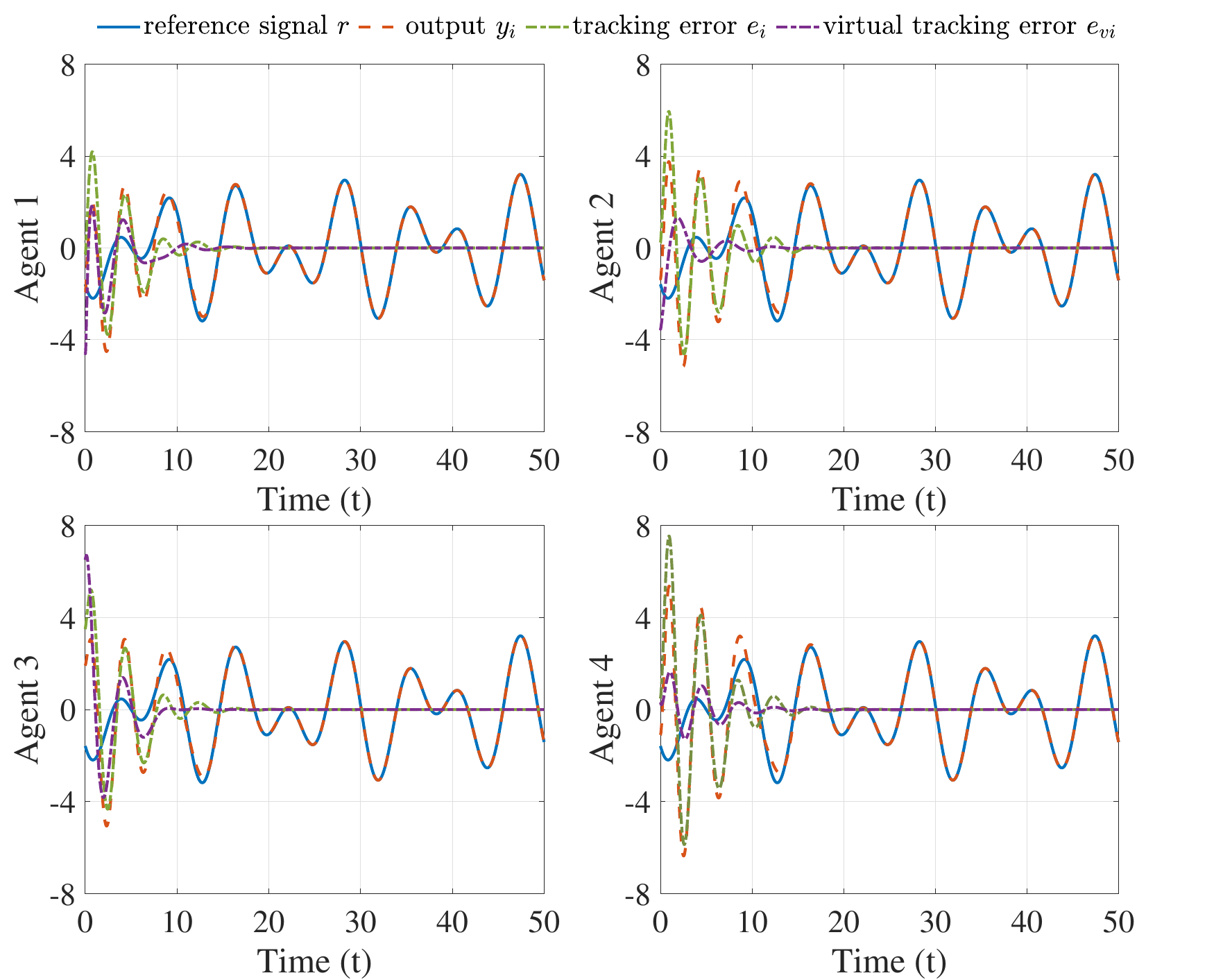}\\
\caption{Tracking performance under the proposed data-driven control approach.}\label{fig:mas:lin}
	\centering
\end{figure}

\begin{figure}[!htb]
	\centering
	\includegraphics[width = 9.3cm]{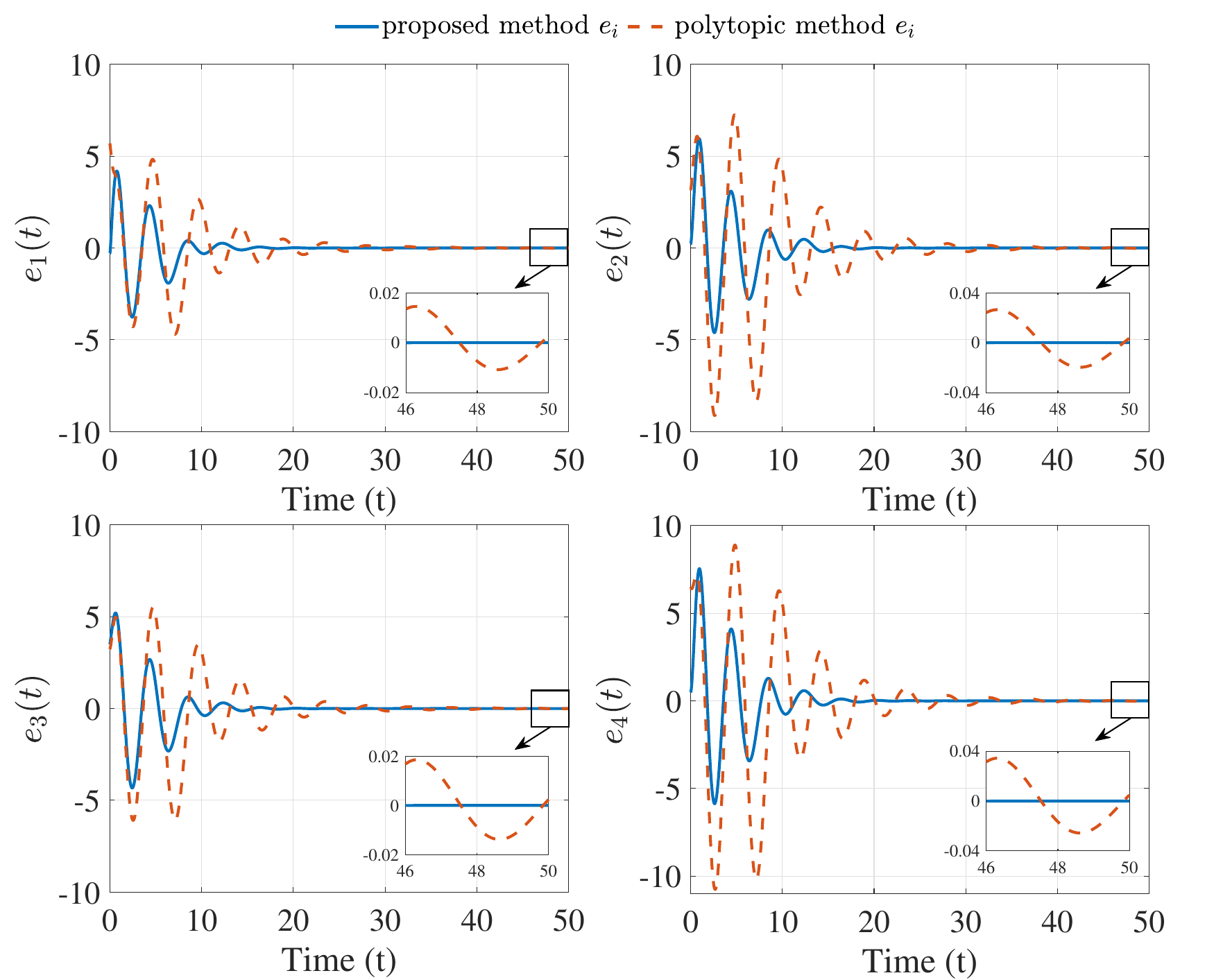}\\
	\caption{Tracking error $e_i(t)$ under the proposed control approach and polytopic control approach \cite{li2024poly}.}\label{fig:mas:linpoly}
	\centering
\end{figure}

\subsection{Nonlinear multi-agent systems}
Consider a continuous-time nonlinear MAS composed of $N$ heterogeneous agents indexed by $1,2,\ldots,N$ and an exosystem indexed by $0$, interacting via a communication network described by a topology $\bar{\mathcal{G}}$. The dynamics of the $i$th agent are modeled by
\begin{subequations} \label{eq:nonmas}
    \begin{align}
     \dot{x}_i  &=f_i(x_i  ,u_i  ,v  ),\\
     y_i   &=c_i(x_i  )
    \end{align}
\end{subequations}
where $f_i$ and $c_i$ are unknown nonlinear functions.
We assume that all the functions in \eqref{eq:nonmas} are sufficiently smooth, satisfying $f_i(0,0,0)=0$ and $c_i(0)=0$.
The exosystem is given by \eqref{eq:mas:exo} and satisfies Assumption~\ref{as:exo:non}.
 
Define the tracking error between the each agent $i$ and the exosystem by
\begin{equation}
    e_i  :=h_i(x_i  ,v  )=y_i  -y_0  .
\end{equation}

Let us rewrite the nonlinear MAS \eqref{eq:nonmas} and the exosystem \eqref{eq:mas:exo} as follows
\begin{subequations}\label{eq:nonmas:linearlized}
    \begin{align}
\dot{x}_i  &=A_ix_i  +B_iu_i  +E_iv  +\alpha_i(x_i,u_i,v)\\
        y_i  &=C_ix_i  +\mu_{i}(x_i)\\
        v  &=Sv  \\
        y_0  &=-Fv
    \end{align}
\end{subequations}
where the functions $\alpha_i(x_i,u_i,v)$ and $\mu_{i}(x_i)$ are the higher-order remainders, and the matrices $A_i$, $B_i$, $C_i$, $E_i$, and $F$ are  defined as
\begin{align*}
    A_i&:=\frac{\partial f_i}{\partial x_i}(0,0,0), ~B_i:=\frac{\partial f_i}{\partial u_i}(0,0,0)\\ C_i&:=\frac{\partial c_i}{\partial x_i}(0),~E_i:=\frac{\partial f_i}{\partial v}(0,0,0).
\end{align*}

Following the results in Section~\ref{sec:single:non}, we implement a linear distributed state feedback control protocol incorporating an internal model for each agent $i=1,2,\ldots,N$ of the form \eqref{eq:controller}.
Formally, we state the problem to be addressed below.

\begin{problem} 
\label{pro:coop:non}{\textbf{\textup{($k$th-order nonlinear cooperative output regulation)}}}
For the nonlinear MAS \eqref{eq:mas}, the exosystem \eqref{eq:mas:exo}, as well as the diagraph $\bar{\mathcal{G}}$, design a distributed control law of the form \eqref{eq:controller} such that for all sufficiently small $x_i(0)$ and $v(0)$, the tracking errors $e_i$ satisfy
\begin{align}
\lim_{t \rightarrow \infty}& (e_i(t)   - O(v^{(k+1)}(t))\nonumber\\
&= \lim_{t \rightarrow \infty} (h_i(x_i(t),v(t)) - O(v^{(k+1)}(t)) = 0
\end{align}
for all $i=1,2,\ldots,N$, where $O(v^{(k+1)})$ is characterized as \eqref{eq:property:o}.
\end{problem}

To solve Problem \ref{pro:coop:non}, we begin by considering the compact form of the multi-agent closed-loop system 
\begin{subequations}
	\begin{align}
		\dot x   &=  \tilde Ax  +\tilde B u  + \tilde E v  +\alpha(x,u,v)\\
		\dot v   &= S v  \\
		e_v   &= (H\otimes I_p)\tilde{C}x  +\tilde F v  +\gamma(x,v)\\
  \dot {z}   &= (I_N\otimes G_1)z  + (I_N\otimes G_2)e_{v}  
	\end{align}
\end{subequations}
where $\alpha= {\rm col}(\alpha_{1},\ldots,\alpha_{N})$ and $\gamma={\rm col}(\gamma_{1},\ldots,\gamma_{N})$. The remaining symbols are defined as in
\eqref{eq:mas:close:x}.
Similar to the linear case, letting $\xi  := {\rm col}(x,z )$, the resulting closed-loop composite system can be written as
\begin{subequations}\label{eq:nonmas:closexi}
    \begin{align}
       \dot \xi	& = {\left[
\begin{matrix}
    \tilde A & 0\\
    (H\otimes G_2) \tilde C & I_N\otimes G_1
\end{matrix}
\right]} \xi + {\left[
\begin{matrix}
   \tilde B\\
    0
\end{matrix}
\right]} u \nonumber\\
&\quad+ {\left[
\begin{matrix}
    \tilde E\\
    (I_N\otimes G_2 )\tilde F
\end{matrix}
\right]} v+ {\left[
\begin{matrix}
    \alpha(x,u,v) \\
    (I_N\otimes G_2) \gamma(x,v)
\end{matrix}
\right]}\nonumber\\
& = \tilde A_{\xi} \xi + \tilde E_\xi v + \psi(\xi,v)\\
	\dot{v} & = S v\\
	e_v   &= (H\otimes I_p)\tilde{C}x  +\tilde F v  +\gamma(x,v).
    \end{align}
\end{subequations}

Following the same step as in Section \ref{sec:single:non}, as long as the controller gain matrices $K_{\xi i}$ are designed such that $\tilde A_\xi$ is Hurwitz stable, Problem \ref{pro:coop:non} is solved.
Therefore, the objective here is to provide sufficient conditions for the design of $K_{\xi i}$ from data.
To this end, let $\mathbb{D}_i := \{(u_i(\tau),{\xi}_i(\tau),\dot{\xi}_i(\tau)\}_{\tau = 0}^{T-1}$
be the data resulting from an experiment carried out on the nonlinear system
\begin{subequations}\label{eq:nonmas:offline}
   \begin{align}
    \dot{\xi}_{i}  &= {A}_{\xi i} \xi_{i}  +{B}_{\xi i} u_{i}  +E_{\xi i}v  +\psi_{i}(\xi_i, v)\\
  \dot v  & = S  v  .
\end{align} 
\end{subequations}
Bearing in mind the analysis of Section \ref{sec:single:non:data},  we define the matrices of unknown exosignal $v  $ and approximating error $\psi_{i}(\xi_i, v)$ by 
\begin{align*}
	D_i &= \left[\begin{matrix}
		\psi_{i}(\xi_i(0),v(0)) & \cdots & \psi_{i}(\xi_i(T-1),v(T-1))
	\end{matrix}\right]\\
	V &= \left[\begin{matrix}
		{v}(0) & {v}(1) & \cdots & {v}(T-1)
	\end{matrix}\right].
\end{align*}
In addition, the data matrices $U_{i-}$, $\Xi_{i-}$, and $\Xi_{i+}$ are defined in \eqref{eq:mas:datamatrix} and they satisfy the identity
\begin{equation}
	\Xi_{i+} = A_{\xi i} \Xi_{i-} + B_{\xi i} U_{i-} + E_{\xi i} V + D_i.
\end{equation}

Assume that sequences $D_i$ and $V$ have bounded energy,  i.e., Assumption~\ref{as:noise} holds in this part as well.
Now, we are in the position to establish our main result using linear data-driven cooperative output regulation theory.

\begin{theorem}
     Consider the MAS \eqref{eq:nonmas}, the exosystem \eqref{eq:mas:exo}, and the graph $\bar{\mathcal{G}}$ under Assumptions ~\ref{as:exo:non} and \ref{as:graph}. 
 For data $U_{i-}$, $\Xi_{i-}$, $\Xi_{i+}$ in \eqref{eq:mas:datamatrix} satisfying Assumptions \ref{as:rank} and \ref{as:noise} and $\Psi_i$, $\Upsilon_i$, $\Sigma_i$ in \eqref{eq:mas:elip}, if the LMIs in \eqref{eq:mas:sdp1} are feasible for some matrices $P_i \succ 0$ and $Y_i$, $i=1,2,\ldots,N$, then the distributed control protocol \eqref{eq:controller} with matrix $K_{\xi i}:= Y_iP_i^{-1}/\lambda_i$ solves Problem \ref{pro:coop:non}.
\end{theorem}
\begin{proof}
    The proof is completed by Theorems~\ref{thm:non} and \ref{thm:mas:linear}.
\end{proof}

\emph{Example 4:} Consider a multi-agent system (MAS) consisting of four ball and beam systems as described in Section~\ref{sec:single:non:exm}, with each system treated as agent $i$ for $i=1,\ldots,4$.
In the heterogeneous case, the parameter $h_0$ is different for each agent, with values of $0.7134$, $0.75$, $0.7647$, and $0.7776$, respectively.
The communication graph $\bar{\mathcal{G}}$ is shown in Fig.~\ref{fig:graph}.

We collected trajectories of length 
 $T = 100$ for each agent from random initial conditions and random inputs uniformly generated from $[-0.1,0.1]$, the exosignal $v$ from $[-0.002,0.002]$, and  noise $d_i$ from $[-0.002,0.002]$ for all $i=1,\ldots,N$.
Under the proposed distributed control protocol \eqref{eq:controller}, Figs.~\ref{fig:mas:non1} and \ref{fig:mas:non2} depict the nonlinear tracking performance of each agent when implementing $1$st-order and $2$nd-order internal models, respectively.
Taking a $2$nd-order internal model as an example, we plot the local stability of the tracking error $e_i(t)$ under the proposed control protocol \eqref{eq:controller} in Fig.~\ref{fig:mas:non3}.
These figures demonstrate the effectiveness of the proposed method in addressing the cooperative output regulation problem of nonlinear MASs.

\begin{figure}[!htb]
	\centering
	\includegraphics[width = 9.5cm]{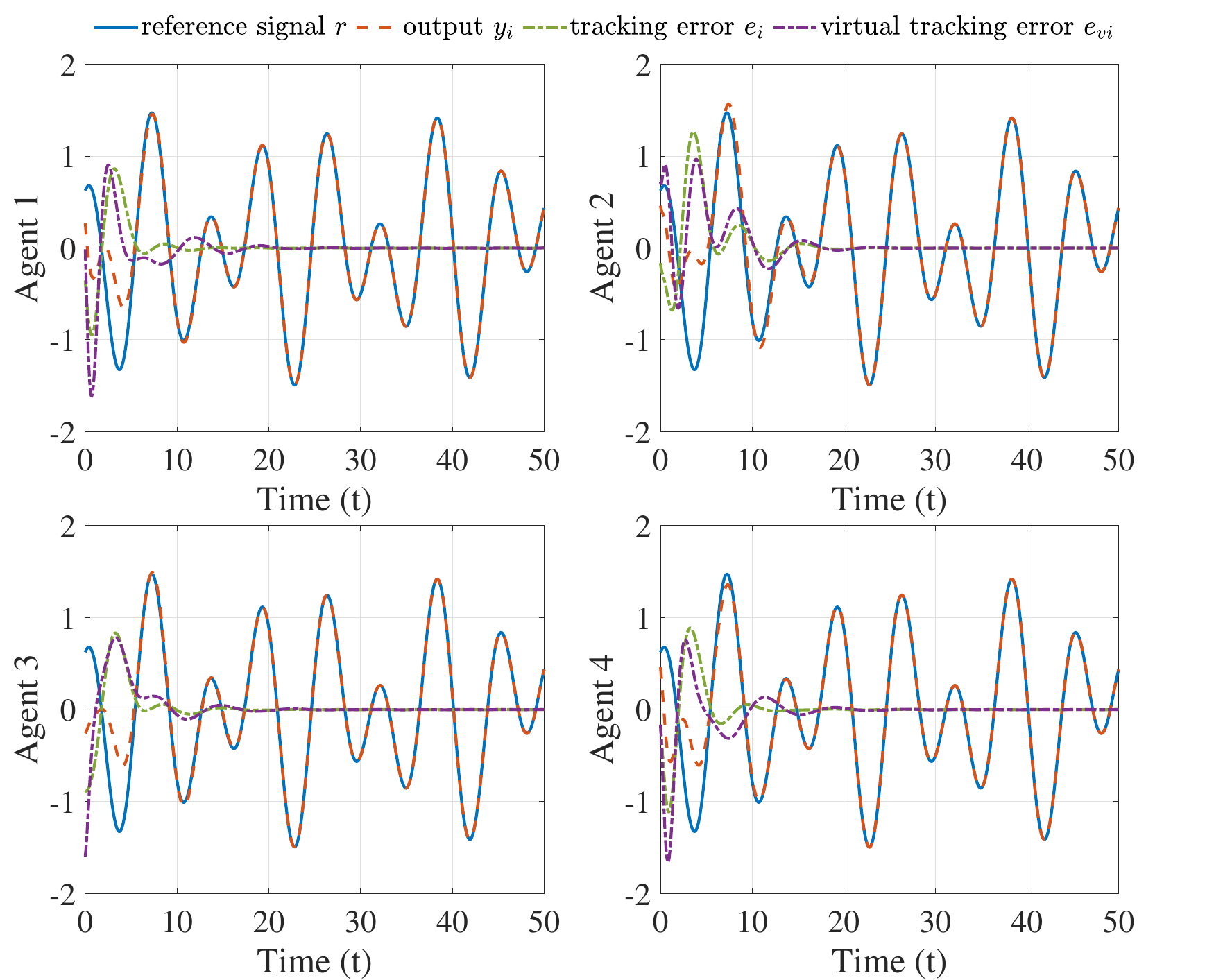}\\
	\caption{Tracking performance of the nonlinear MAS with $k=1$.}\label{fig:mas:non1}
	\centering
\end{figure}

\begin{figure}[!htb]
	\centering
	\includegraphics[width = 9.5cm]{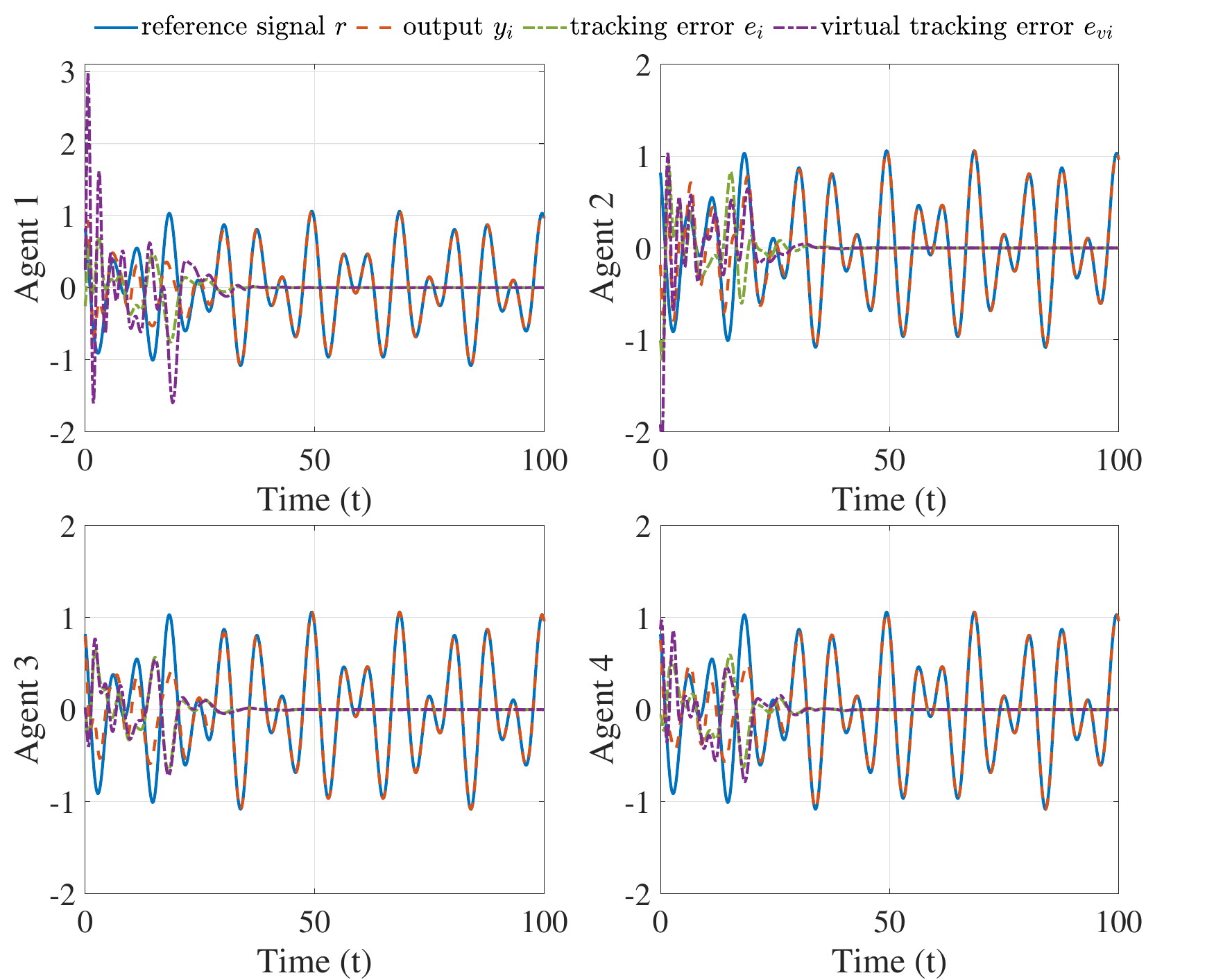}\\
	\caption{Tracking performance of the nonlinear MAS with $k=2$.}\label{fig:mas:non2}
	\centering
\end{figure}

\begin{figure}[!htb]
	\centering
	\includegraphics[width = 9cm]{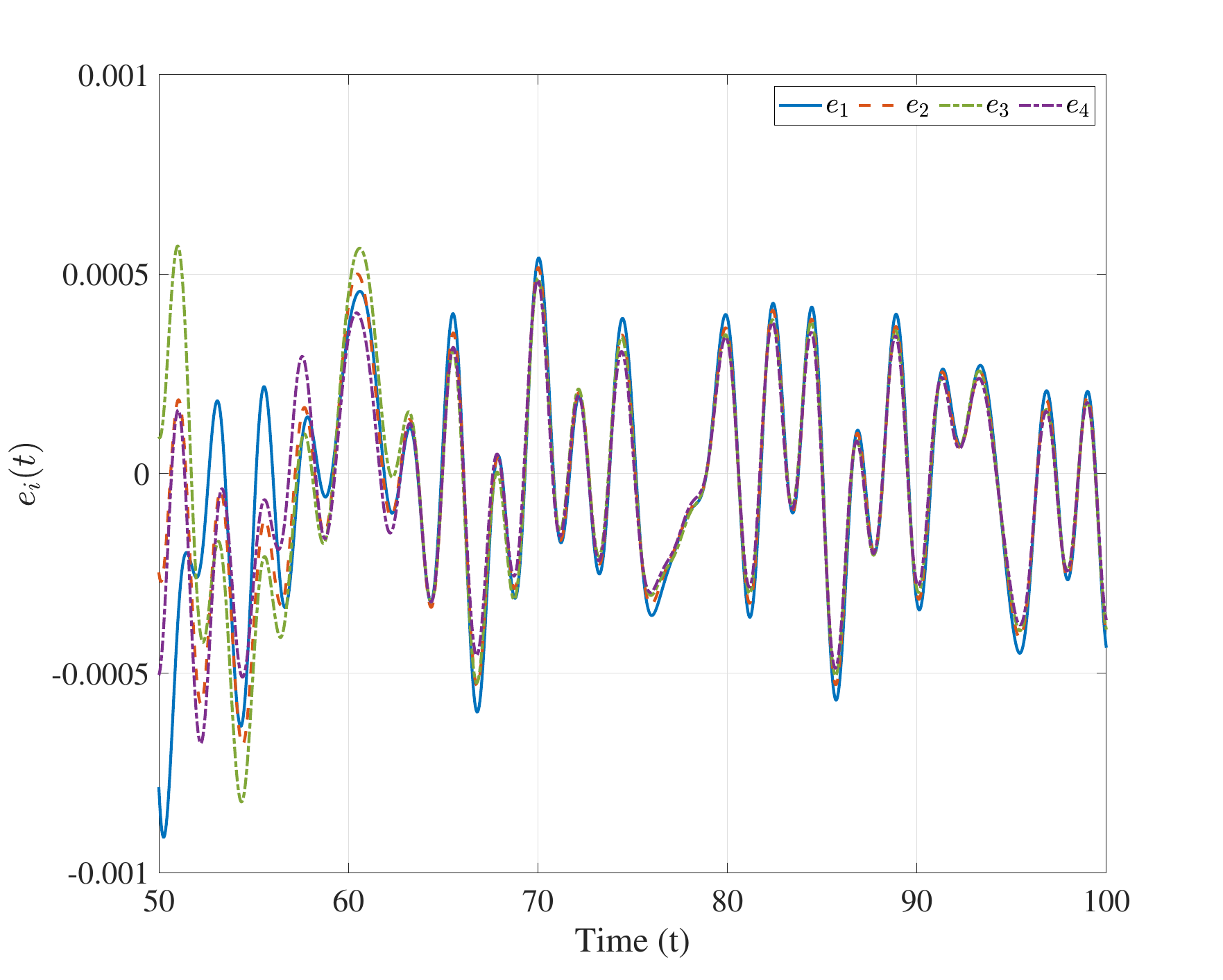}\\
	\caption{Tracking error of the nonlinear MAS with $k=2$.}\label{fig:mas:non3}
	\centering
\end{figure}

\section{Conclusion}\label{sec:conclusion}
This paper addressed the problem of output regulation for both unknown linear and nonlinear, single and multi-agent systems (MASs) using noisy data. Departing from the traditional approach of solving data-based output regulation equations, we proposed a data-driven internal model-based controller. This controller can be designed by solving a simple and low-complexity data-based linear matrix inequality (LMI). The proposed method is proven to be effective for nonlinear systems, achieving 
$k$th-order output regulation. Furthermore, the approach extends seamlessly to MASs. Numerical examples have demonstrated the efficacy and robustness of the proposed data-driven control strategy.

The results of this paper demonstrate that our approach is promising for nonlinear systems, though we have only scratched the surface of this research area. Exploring other methods for handling nonlinearity, as discussed in Remark \ref{rmk:nonctrl}, or considering other types of Lyapunov functions, e.g., polynomial Lyapunov functions in \cite{guo2022data}, to achieve less conservative results are all interesting directions for future work.

\bibliographystyle{IEEEtran}
\bibliography{ddOutReg}

\end{document}